\documentclass[submission,copyright,creativecommons]{eptcs}
\providecommand{\event}{AiML 2026} 

\usepackage{aiml26}

\usepackage{iftex}

\usepackage{hyperref}
\usepackage{amssymb}
\usepackage{amsfonts}

\ifpdf
  \usepackage{underscore}         %
  \usepackage[T1]{fontenc}        %
\else
  \usepackage{breakurl}           %
\fi

\usepackage{graphicx}

\usepackage{hhline,multirow,tabularx,tikz,xstring,amsmath}
\usepackage{array}

\usepackage{amsfonts,mathrsfs,epic,amsbsy,amstext, amssymb, amscd,amsmath, textcomp}

\usepackage[all]{xy}

\usepackage{todonotes} %
\usepackage{comment} 

\newif\ifshow
\showfalse 

\ifshow
  \includecomment{rem}
\else
  \excludecomment{rem} 
\fi

\renewcommand{\phi}{\varphi}

\newcommand{\bq}{\begin{quote}}
\newcommand{\eq}{\end{quote}}

\newtheorem{Th}{Theorem}
\newtheorem{ax}{Axiom}
\newtheorem{df}{Definition}
\newtheorem{pr}{Proposition}
\newtheorem{cl}{Corollary}
\newtheorem{re}{Remark}
\newtheorem{as}{Assumption}
\newtheorem{wg}{Wild Guess}
\newtheorem{ex}{Example}

\newcommand{\Agents}{\mathcal{A}}

\newcommand{\ot}{\leftarrow}

\newcommand{\bM}{{\mathbf M}}

\newcommand{\bS}{{\mathbf S}}

\newcommand{\bE}{{\mathbf E}}
\newcommand{\be}{{\mathbf e}}

\newcommand{\boldf}{{\mathbf f}}
\newcommand{\myComment}[1]{}

\newcommand{\us}{\underline{{s}}}
\newcommand{\ue
}{\underline{{e}}}

\newcommand{\ux}{\overline{{x}}}
\newcommand{\uy}{\overline{{y}}}
\newcommand{\uz}{\overline{{z}}}

\newcommand{\bth}{\begin{Th}}
\newcommand{\Eth}{\end{Th}}
\newcommand{\bax}{\begin{ax}}
\newcommand{\eax}{\end{ax}}
\newcommand{\blm}{\begin{lm}}
\newcommand{\elm}{\end{lm}}
\newcommand{\bdf}{\begin{df}}
\newcommand{\edf}{\end{df}}
\newcommand{\bpr}{\begin{pr}}
\newcommand{\epr}{\end{pr}}
\newcommand{\bcl}{\begin{cl}}
\newcommand{\ecl}{\end{cl}}
\newcommand{\bre}{\begin{re}}
\newcommand{\ere}{\end{re}}
\newcommand{\bas}{\begin{as}}
\newcommand{\eas}{\end{as}}
\newcommand{\bwg}{\begin{wg}}
\newcommand{\ewg}{\end{wg}}
\newcommand{\bex}{\begin{ex}}
\newcommand{\eex}{\end{ex}}

\newcommand{\bit}{\begin{itemize}}
\newcommand{\eit}{\end{itemize}\par\noindent} \newcommand{\beq}{\begin{equation}}
\newcommand{\eeq}{\end{equation}\par\noindent} \newcommand{\beqa}{\begin{eqnarray*}}
\newcommand{\eeqa}{\end{eqnarray*}\par\noindent} \newcommand{\beqn}{\begin{eqnarray}}
\newcommand{\eeqn}{\end{eqnarray}\par\noindent}

\newcommand{\sima}{\sim_a}

\newcommand{\simb}{\sim_b}

\newcommand{\sime}{\sim_e}

\newcommand{\simB}{\sim_B}

\newcommand{\simA}{\sim_A}

\newcommand{\simc}{\sim_c}

\newcommand{\bD}{\mathbf{D}}

\title{The Logic of Data Access and Data Exchanges}
\author{Alexandru Baltag
\institute{University of Amsterdam, ILLC\\ Amsterdam, Netherlands}
\email{thealexandrubaltag@gmail.com}
\and
 Sonja Smets
\institute{University of Amsterdam, ILLC\\ Amsterdam, Netherlands}
\email{\quad s.j.l.smets@uva.nl}
}

\newcommand{\titlerunning}{The Logic of Data Access and Data Exchanges}
\newcommand{\authorrunning}{A. Baltag, S. Smets}

\hypersetup{
  bookmarksnumbered,
  pdftitle    = {\titlerunning},
  pdfauthor   = {\authorrunning},
  pdfsubject  = {EPTCS},               
  pdfkeywords = {Dynamic Epistemic Logic,
Epistemic Logic,
Knowledge of values,
Data Exchange Events,
Group knowledge,
Common Distributed Knowledge,
Logics for Multi-Agent Systems}
}

\begin{document}

\maketitle

\begin{abstract}
We investigate a new logic that extends Dynamic Epistemic Logic (DEL), by combining standard epistemic modalities $K_a \varphi$ and $K_A \varphi$ for (individual and distributed) \emph{propositional knowledge} with operators $K_a^\varphi x$, $K_A^\varphi x$ denoting \emph{(conditional) non-propositional knowledge of a number} (in which an agent $a$ or a group $A$ have knowledge of the value of some variable $x$, if given additional information $\varphi$). We also generalize these operators, by considering formulas $|x|_A^\varphi \leq N$  (for any natural number $N$), saying that: \emph{conditional on $\varphi$, $A$ can narrow down the possible values of variable $x$ to at most $N$ possibilities}. In order to name and compare such hypothetical values, we extend the logic further with \emph{definite descriptions based on minimization operators}: $\mu_N x_A^\varphi$ denotes the \emph{least of the $N$ possible values of $x$} (according to some fixed order $\leq$) that are considered possible by group $A$ (given condition $\varphi$).
On this static base, we consider DEL-style extensions with \emph{dynamic modalities for general `data-exchange events'} (covering private and public propositional announcements, but also secret hacking of a private database, or public sharing of one's data via open-source repositories, etc). In such scenarios, whole `chunks' of information may be exchanged or modified: once access to a given source is gained, all the `data' stored at that 
specific location becomes available. We give complete axiomatizations for the resulting logics, and prove their decidability and co-expressivity. \end{abstract}

\vspace{-2mm}
\section{Introduction}

\vspace{-2mm}

Information may come in both propositional form (Boolean variables) and non-propositional form (e.g. numbers, names, addresses, pictures, videos, etc). Such \emph{data} can be encoded as \emph{values of local variables} (e.g. a cryptographic key or a password), that are stored and processed in certain \emph{locations} or `sites'(e.g., websites, folders, databases, etc),  and can be retrieved when these sites are accessed as `sources'. Each such source can be thought of as an \emph{agent} (either because it actually is the knowledge base of a natural or artificial agent, or because we think of it as an abstract `agent' possessing exactly the
information that is stored at it).\footnote{The same piece of data can be stored/copied at multiple locations, and each location can store multiple pieces of data.} Such agents either `own' their data, or else have gained access to them from other sources: we say that the agents `know' the values of these variables, and some of them can also modify these values. More complex data may have an \emph{extended location}, the information being \emph{distributed} among a number of agents: one can recover the value of such complex variables only by accessing several sources.

To deal with non-propositional information, multi-agent epistemic logic has been extended in recent years with ``knowing-what'' operators $K_a x$ or $K_A x$ for (individual or distributed) \emph{knowledge of (the value of) some variable $x$} (in addition to the traditional modalities $K_a \varphi$ and $K_A \varphi$ for individual or distributed knowledge of a proposition $\varphi$). This work traces back to Plaza \cite{Plaza} and subsequent investigations by \cite{vEGWang,Yanjing,WangFan1,WangFan2,GuWang,Ding,Hong,Baltag2016,BS25} on formalizing `knowledge de re'. While Plaza axiomatized the static logic of `knowing what', the extension with dynamic operators $[!\varphi]\psi$ for (propositional) public announcements $!\varphi$ was only later axiomatized by Wang and Fan \cite{WangFan2}. To `pre-encode' this dynamics using DEL-style reduction axioms \cite{sep-dynamic-epistemic}, these authors needed a \emph{conditional version of `knowing what'}: $K_a^\varphi x$ means that $a$ knows the value of $x$ given the information that proposition $\varphi$ was the case.

In recent work \cite{BS25}, we extended this setting with operators $K_A^\varphi x$ for \emph{conditional distributed knowledge of the value of $x$ by a group $A$}. This was motivated by the need to deal with a more complex dynamics, going beyond propositional announcements or other traditional DEL-events 
\cite{sep-dynamic-epistemic,BM,BMS,DHK,NEW-BMS,BMD,LDII}. In a \emph{data-exchange event}, agents may gain access to `sources' (i.e., to other agents' knowledge bases), in which case they can be assumed to instantly `read' (and copy in their own knowledge base) \emph{all} the information stored at those sources. Such events were previously considered in our work \cite{BS20,BS24}, and in fact they subsume many other forms of non-propositional informational dynamics that were previously considered, e.g. ``tell us all you know'' in \cite{Baltag2010} (and with a different, interrogative interpretation in \cite{BenthemMinica1,BenthemMinica2}),  in-group sharing \cite{Baltag2010,Boddy2014,Goldbach2015,BBS2016}, and `resolution' \cite{AgotnesWang2017}. Other examples of data-exchange events covered in \cite{BS25} include private changing of the value of a variable (e.g. one's password), parallel data-sharing within different subgroups (e.g., in a poster session),  suspected hacking of a private database, private detection of such hacking, etc. As noted in \cite{BS20}, when an agent $a$ gains access to another agent $b$'s knowledge base, her new state of knowledge matches the \emph{distributed knowledge} of the group $\{a, b\}$. So on the side of propositional knowledge, we needed operators $K_A \varphi$ for \emph{distributed knowledge\footnote{The notation $D_A\varphi$ is typically used for distributed knowledge, but we prefer here $K_A\varphi$ because we think of this as capturing a natural notion of (virtual) group knowledge.}  within group $A$}. A similar move was also necessary on the side of ``knowing what'': operators $K_A x$, expressing that \emph{the group $A$ has distributed knowledge of the value of $x$}; in fact, to pre-encode public-announcement dynamics we needed the \emph{conditional version} $K_A^\varphi x$ of these operators.
Note that this type of conditional knowledge may \emph{not} imply knowledge of the `real' value of $x$ (in the actual world), but only the fact that the agent or group \emph{can uniquely determine a hypothetical value} of $x$ (applicable only to worlds satisfying $\varphi$). To axiomatize these notions, in the presence of equality of values $x=y$, we were lead to introduce \emph{definite description terms} $x_A^\varphi$ that can \emph{explicitly refer to such hypothetical values}:
$x_a^\varphi$ denotes (in the actual world) the \emph{(unique) value that $x$ would have according to agent $a$ if $\varphi$ were true}. Using this, we were able to provide in \cite{BS25} a complete axiomatization for the corresponding dynamic-epistemic logic. But for this, we had to \emph{restrict} the dynamics to `semi-public' events (a class that does \emph{not} include e.g. secret hacking). 
On the other hand, in \cite{BS24} we axiomatized a logic with \emph{unrestricted dynamics}, covering `arbitrary' data-exchange events (including complicated hacking scenarios). But this was based on a \emph{restricted static logic base}, that did \emph{not} include non-propositional knowledge operators, but only the traditional epistemic modalities $K_a \varphi$ and $K_A\varphi$, as well as (a major generalization of) common knowledge $C_A \varphi$ (in the form of polyadic conditional-epistemic group modalities $C_A^\be \varphi_1 \ldots \varphi_n$, indexed by data-exchange events $\be$). There were \emph{no explicit variables} $x$ in the logic, so no assertions $x=y$ or $Px_1\ldots x_n$ talking about the properties of specific pieces of (non-propositional) data, as well as \emph{no} ``knowledge-what'' operators $K_a x$, $K_A x$ or $K_A^\varphi x$. And, because of these limitations, the class of data-exchange events was still inherently restricted: e.g., there was no event of privately changing the value (such as one's password), and no events involving preconditions of the form $K_a x$ (for instance, no ``conditional hacking'', in which it is common knowledge that an agent $a$ may be hacking another agent $b$'s database if and only if she came to know $b$'s password).

\smallskip

In this paper, we overcome the limitations present in the frameworks of both \cite{BS24} and \cite{BS25}, by providing a common generalization, while also greatly increasing the expressivity of the logic. More specifically, we generalize the operators $K_A ^\varphi x$ to expressions $|x|_A^\varphi\leq N$ (for any natural number $N$), saying that: \emph{if given additional information $\varphi$, the group $A$ can narrow down the possible values of variable $x$ to at most $N$ possibilities} (for any natural number $N$). In other words, the group has distributed knowledge that, if $\varphi$ is the case then the value of $x$ belongs to a given list of $N$ possible values.\footnote{Conditional knowledge $K_A^\varphi x$ of the value can now be recovered as a special case (for $N=1$).} This is a useful addition when dealing with cryptographic protocols. Indeed, if $N$ is `small enough' and $x$ is say another agent $b$'s communication key or private password, then any intruder $a$ having the capability $|x|_a^\varphi \leq N$ will be able to hack $b$'s communication (by simply trying out all the $N$ possible values), as soon as she may receive the information $\varphi$. The same applies to a group $A$ of hackers if they can collectively narrow down the possibilities, i.e., have the capability $|x|_A^\varphi \leq N$. 
As done in \cite{BS25} for $K_A^\varphi x$, whenever we have  $|x|_a \leq N$ we introduce \emph{definite descriptions denoting each of the $N$ hypothetical values of $x$}. For this, we assume given a salient \emph{total order} $\leq$ on the set of values. This is useful in numerical applications, but it also
allows us to convert cardinality statements $|x|_a \leq N$ into ordinal descriptions $1_N.x_A^\varphi$, $2_N.x_A^\varphi$ etc, denoting ``the first (or the second, etc) of the $\leq N$ possible values of $x$ according to $A$ given $\varphi$''.\footnote{We can define all these using a \emph{cutoff-minimization construct} $\mu_N x_A^\varphi$, which is the same as $1_N.x_A^\varphi$: it denotes \emph{the least of the $\leq N$ possible values of $x$ (according to $A$ given $\varphi$)}.}
To capture the properties of the total order, we include in our language \emph{order statements $x\leq y$} (besides equality of values $x=y$ and possibly other predicates $P x_1\ldots x_n$). We obtain a very expressive dynamic-epistemic logic, and we 
are able to completely axiomatize it and prove its decidability. 
Our proofs make an innovative use of known methods. For the static logic, we first use filtration to obtain a quasi-model: this is not a `real' model, but just a syntactic construction, in which \emph{variables have no values} (and the group relations are non-standard, as they are \emph{not intersections} of the individual relations). We then follow the usual method for dealing with distributed knowledge, by unraveling the model into an infinite tree and redefining the relations to make them standard. However, new complications ensue due to the variables, and especially due to the combination of equality $x=y$ and non-propositional knowledge $K_A x$; the problems are in fact compounded by the presence of formulas for order $x\leq y$ and narrowing down $|x|_A^\varphi\leq N$. Here, we use our definite descriptions to ensure that the local properties of variables $x$ are preserved and continuously transmitted at far-away nodes of the tree.\footnote{
In effect, the `values' of variables $x$ ``hitch a ride'' on the back of various group epistemic transitions $\to_A$ in the tree, and their properties are transmitted by passing them to \emph{other ($A$-preserved) terms}.} Defining the total order on values over the tree requires a use of the Order Extension Principle, and hence of the Axiom of Choice.
Finally, the completeness for the dynamic logic uses DEL-style reduction axioms; but again, the details are quite tantalizing. Even the soundness of some reduction laws (e.g., the `Cutoff-Min. Change' axiom) is not at all obvious!

We should mention a self-imposed limitation: due to the page-limit, we chose \emph{not} to include common knowledge operators $C_A \varphi$, whose dynamics 
introduces another source of complexity, requiring additional technical developments and proofs. We leave this step for the journal-version of this paper.
\vspace{-4mm}
\section{Motivating Examples}
\vspace{-3mm}

\par\noindent\textbf{Example 1} Alice and Bob are each given a number $x(a), x(b)\in \mathbb{N}=\{0, 1, 2,\ldots\}$, while the value of the sum $x(e):=x(a)+ x(b)$ is stored in a closed envelope $e$. It is common knowledge that:
(1) \textit{each of the two sees his/her own number, but neither of them can see the others' number, nor the sum $x(e)$};
(2) \textit{one of the two numbers $x^a, x^b$ is the immediate successor of the other}
(i.e. either $x(a)= x(b) + 1$ or $x(b) = x(a) + 1$). We represent this situation using an \emph{infinite Kripke model}, where states are triplets of numbers $(x(a), x(b), x(e))$, and the accessibility relations $\sima$ and $\simb$ encode the agents' uncertainty. We also have an accessibility relation $\sime$ for the `envelope' (treated as an abstract `agent'), encoding the information contained in it. The formula $K_a x(a) \wedge K_b x(b)$, saying  
that \emph{agents know their numbers}, is valid on this model, while $K_e x(e)$ says that the sum $x(e)$ is stored in the envelope. On the other hand, agents $a$ and $b$ can \emph{together} figure out the sum: we have $K_{\{a,b\}} x(e)$, i.e.
the \emph{group} $A=\{a,b\}$ has \textit{``distributed'' knowledge} of $x(e)$. 
The formulas $|x(a)|_b \leq 2$, $|x(b)|_a\leq 2$, $|x(e)|_a\leq 2$ and $|x(e)|_b\leq 2$ say that the live agents can \textit{narrow down the possibilities} (for the other's number, as well as for the sum) \textit{to at most two}. To \emph{name} these values, we write e.g. $1_2.x(e)_a$ for the \emph{first} (i.e., the \emph{least}) of the \emph{two} possible values of the sum $x(e)$  according to Alice, $2_2.x(e)_a$ for the \emph{second} least (=largest) of the \emph{two}, etc. The picture is:
\medskip\par\noindent

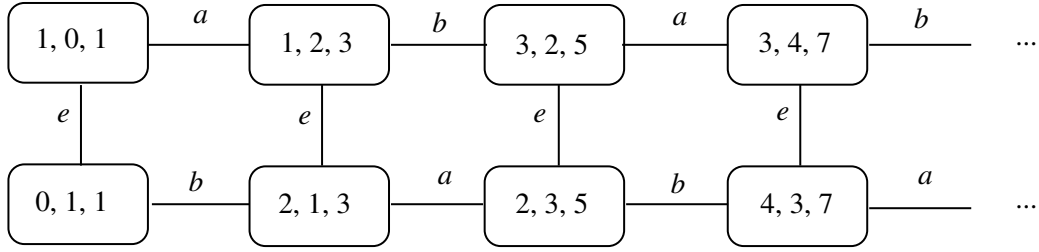
\tikzset{every picture/.style={line width=0.75pt}} 

\begin{tikzpicture}[x=0.75pt,y=0.75pt,yscale=-1,xscale=1]
\draw   (114,60) .. controls (114,55.58) and (117.58,52) .. (122,52) -- (176,52) .. controls (180.42,52) and (184,55.58) .. (184,60) -- (184,84) .. controls (184,88.42) and (180.42,92) .. (176,92) -- (122,92) .. controls (117.58,92) and (114,88.42) .. (114,84) -- cycle ;
\draw   (114,141) .. controls (114,136.58) and (117.58,133) .. (122,133) -- (176,133) .. controls (180.42,133) and (184,136.58) .. (184,141) -- (184,165) .. controls (184,169.42) and (180.42,173) .. (176,173) -- (122,173) .. controls (117.58,173) and (114,169.42) .. (114,165) -- cycle ;
\draw   (235,61) .. controls (235,56.58) and (238.58,53) .. (243,53) -- (297,53) .. controls (301.42,53) and (305,56.58) .. (305,61) -- (305,85) .. controls (305,89.42) and (301.42,93) .. (297,93) -- (243,93) .. controls (238.58,93) and (235,89.42) .. (235,85) -- cycle ;
\draw   (235,142) .. controls (235,137.58) and (238.58,134) .. (243,134) -- (297,134) .. controls (301.42,134) and (305,137.58) .. (305,142) -- (305,166) .. controls (305,170.42) and (301.42,174) .. (297,174) -- (243,174) .. controls (238.58,174) and (235,170.42) .. (235,166) -- cycle ;
\draw   (353,61) .. controls (353,56.58) and (356.58,53) .. (361,53) -- (415,53) .. controls (419.42,53) and (423,56.58) .. (423,61) -- (423,85) .. controls (423,89.42) and (419.42,93) .. (415,93) -- (361,93) .. controls (356.58,93) and (353,89.42) .. (353,85) -- cycle ;
\draw   (353,142) .. controls (353,137.58) and (356.58,134) .. (361,134) -- (415,134) .. controls (419.42,134) and (423,137.58) .. (423,142) -- (423,166) .. controls (423,170.42) and (419.42,174) .. (415,174) -- (361,174) .. controls (356.58,174) and (353,170.42) .. (353,166) -- cycle ;
\draw   (475,61) .. controls (475,56.58) and (478.58,53) .. (483,53) -- (537,53) .. controls (541.42,53) and (545,56.58) .. (545,61) -- (545,85) .. controls (545,89.42) and (541.42,93) .. (537,93) -- (483,93) .. controls (478.58,93) and (475,89.42) .. (475,85) -- cycle ;
\draw   (475,142) .. controls (475,137.58) and (478.58,134) .. (483,134) -- (537,134) .. controls (541.42,134) and (545,137.58) .. (545,142) -- (545,166) .. controls (545,170.42) and (541.42,174) .. (537,174) -- (483,174) .. controls (478.58,174) and (475,170.42) .. (475,166) -- cycle ;
\draw    (184,73) -- (235.2,73) ;
\draw    (306,73) -- (354.2,73) ;
\draw    (422,73) -- (475.2,73) ;
\draw    (186,153) -- (235.2,153) ;
\draw    (306,153) -- (352.2,153) ;
\draw    (424,153) -- (475.2,153) ;
\draw    (150.4,93.2) -- (150.4,133.6) ;
\draw    (271.4,93.2) -- (271.4,133.6) ;
\draw    (390.2,93.2) -- (390.2,133.6) ;
\draw    (511.4,93.2) -- (511.4,133.6) ;
\draw    (546,73) -- (597.2,73) ;
\draw    (546,155) -- (597.2,155) ;

\draw (126,145) node [anchor=north west][inner sep=0.75pt]   [align=left] {0, 1, 1};
\draw (126,63) node [anchor=north west][inner sep=0.75pt]   [align=left] {1, 0, 1};
\draw (248,146) node [anchor=north west][inner sep=0.75pt]   [align=left] {2, 1, 3};
\draw (367,146) node [anchor=north west][inner sep=0.75pt]   [align=left] {2, 3, 5};
\draw (490,146) node [anchor=north west][inner sep=0.75pt]   [align=left] {4, 3, 7};
\draw (249,65) node [anchor=north west][inner sep=0.75pt]   [align=left] {1, 2, 3};
\draw (367,66) node [anchor=north west][inner sep=0.75pt]   [align=left] {3, 2, 5};
\draw (489,65) node [anchor=north west][inner sep=0.75pt]   [align=left] {3, 4, 7};
\draw (198,54) node [anchor=north west][inner sep=0.75pt]   [align=left] {\,\,\,$a$};
\draw (318,55) node [anchor=north west][inner sep=0.75pt]   [align=left] {\,\,\,$b$};
\draw (446,55) node [anchor=north west][inner sep=0.75pt]   [align=left] {$a$};
\draw (567,53) node [anchor=north west][inner sep=0.75pt]   [align=left] {$b$};
\draw (328,135) node [anchor=north west][inner sep=0.75pt]   [align=left] {$a$};
\draw (569,135) node [anchor=north west][inner sep=0.75pt]   [align=left] {$a$};
\draw (203,135) node [anchor=north west][inner sep=0.75pt]   [align=left] {$b$};
\draw (445,136) node [anchor=north west][inner sep=0.75pt]   [align=left] {$b$};
\draw (137,104) node [anchor=north west][inner sep=0.75pt]   [align=left] {$e$};
\draw (258,105) node [anchor=north west][inner sep=0.75pt]   [align=left] {$e$};
\draw (376,104) node [anchor=north west][inner sep=0.75pt]   [align=left] {$e$};
\draw (498,103) node [anchor=north west][inner sep=0.75pt]   [align=left] {$e$};
\draw (618,70) node [anchor=north west][inner sep=0.75pt]   [align=left] {...};
\draw (618,152) node [anchor=north west][inner sep=0.75pt]   [align=left] {...};
\end{tikzpicture}

\smallskip\par\noindent\textbf{Example 1, continued}
Next, it is common knowledge that \textit{Alice 'hacks' Bob's database}. 
This is a `semi-public' event $!(a:b)$, which \emph{updates the model} by replacing $\sima$  with the intersection $\sima\cap\simb$:
\medskip\par\noindent

{\tikzset{every picture/.style={line width=0.75pt}}        

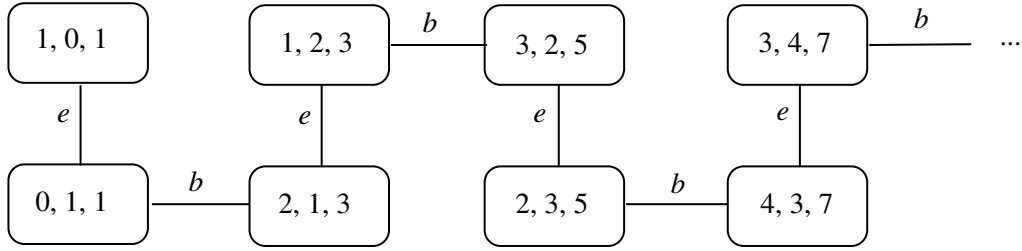
\begin{tikzpicture}[x=0.75pt,y=0.75pt,yscale=-1,xscale=1]
%

\draw   (134,80) .. controls (134,75.58) and (137.58,72) .. (142,72) -- (196,72) .. controls (200.42,72) and (204,75.58) .. (204,80) -- (204,104) .. controls (204,108.42) and (200.42,112) .. (196,112) -- (142,112) .. controls (137.58,112) and (134,108.42) .. (134,104) -- cycle ;
\draw   (134,161) .. controls (134,156.58) and (137.58,153) .. (142,153) -- (196,153) .. controls (200.42,153) and (204,156.58) .. (204,161) -- (204,185) .. controls (204,189.42) and (200.42,193) .. (196,193) -- (142,193) .. controls (137.58,193) and (134,189.42) .. (134,185) -- cycle ;
\draw   (255,81) .. controls (255,76.58) and (258.58,73) .. (263,73) -- (317,73) .. controls (321.42,73) and (325,76.58) .. (325,81) -- (325,105) .. controls (325,109.42) and (321.42,113) .. (317,113) -- (263,113) .. controls (258.58,113) and (255,109.42) .. (255,105) -- cycle ;
\draw   (255,162) .. controls (255,157.58) and (258.58,154) .. (263,154) -- (317,154) .. controls (321.42,154) and (325,157.58) .. (325,162) -- (325,186) .. controls (325,190.42) and (321.42,194) .. (317,194) -- (263,194) .. controls (258.58,194) and (255,190.42) .. (255,186) -- cycle ;
\draw   (373,81) .. controls (373,76.58) and (376.58,73) .. (381,73) -- (435,73) .. controls (439.42,73) and (443,76.58) .. (443,81) -- (443,105) .. controls (443,109.42) and (439.42,113) .. (435,113) -- (381,113) .. controls (376.58,113) and (373,109.42) .. (373,105) -- cycle ;
\draw   (373,162) .. controls (373,157.58) and (376.58,154) .. (381,154) -- (435,154) .. controls (439.42,154) and (443,157.58) .. (443,162) -- (443,186) .. controls (443,190.42) and (439.42,194) .. (435,194) -- (381,194) .. controls (376.58,194) and (373,190.42) .. (373,186) -- cycle ;
\draw   (495,81) .. controls (495,76.58) and (498.58,73) .. (503,73) -- (557,73) .. controls (561.42,73) and (565,76.58) .. (565,81) -- (565,105) .. controls (565,109.42) and (561.42,113) .. (557,113) -- (503,113) .. controls (498.58,113) and (495,109.42) .. (495,105) -- cycle ;
\draw   (495,162) .. controls (495,157.58) and (498.58,154) .. (503,154) -- (557,154) .. controls (561.42,154) and (565,157.58) .. (565,162) -- (565,186) .. controls (565,190.42) and (561.42,194) .. (557,194) -- (503,194) .. controls (498.58,194) and (495,190.42) .. (495,186) -- cycle ;
\draw    (326,93) -- (374.2,93) ;
\draw    (206,173) -- (255.2,173) ;
\draw    (444,173) -- (495.2,173) ;
\draw    (170.2,113.2) -- (170.2,153.6) ;
\draw    (291.2,113.2) -- (291.2,153.6) ;
\draw    (410.2,113.2) -- (410.2,153.6) ;
\draw    (531.2,113.2) -- (531.2,153.6) ;
\draw    (566,93) -- (617.2,92.6) ;

\draw (146,165) node [anchor=north west][inner sep=0.75pt]   [align=left] {0, 1, 1};
\draw (146,83) node [anchor=north west][inner sep=0.75pt]   [align=left] {1, 0, 1};
\draw (268,166) node [anchor=north west][inner sep=0.75pt]   [align=left] {2, 1, 3};
\draw (387,166) node [anchor=north west][inner sep=0.75pt]   [align=left] {2, 3, 5};
\draw (510,166) node [anchor=north west][inner sep=0.75pt]   [align=left] {4, 3, 7};
\draw (269,85) node [anchor=north west][inner sep=0.75pt]   [align=left] {1, 2, 3};
\draw (387,86) node [anchor=north west][inner sep=0.75pt]   [align=left] {3, 2, 5};
\draw (509,85) node [anchor=north west][inner sep=0.75pt]   [align=left] {3, 4, 7};
\draw (338,75) node [anchor=north west][inner sep=0.75pt]   [align=left] {\,$b$};
\draw (587,73) node [anchor=north west][inner sep=0.75pt]   [align=left] {$b$};
\draw (223,155) node [anchor=north west][inner sep=0.75pt]   [align=left] {$b$};
\draw (465,156) node [anchor=north west][inner sep=0.75pt]   [align=left] {$b$};
\draw (157,124) node [anchor=north west][inner sep=0.75pt]   [align=left] {$e$};
\draw (278,125) node [anchor=north west][inner sep=0.75pt]   [align=left] {$e$};
\draw (396,124) node [anchor=north west][inner sep=0.75pt]   [align=left] {$e$};
\draw (518,123) node [anchor=north west][inner sep=0.75pt]   [align=left] {$e$};
\draw (630,90) node [anchor=north west][inner sep=0.75pt]   [align=left] {...};

\end{tikzpicture}
}

\smallskip\par\noindent\textbf{Example 2: an alternative scenario}
Start in the initial situation from Example 1. This time no hacking is allowed, but Alice and Bob are asked: ``\textit{Do you know each other's number}''? 
They both answer (truthfully, publicly and simultaneously): ``\textit{I don't know}''. This is a \emph{public announcement}
$! (\neg K_a x(b) \wedge \neg K_b x(a))$. The \emph{updated model} is obtained by deleting states $(0,1,1)$ and $(1,0,1)$ from the initial model:

\medskip\par\noindent

{\tikzset{every picture/.style={line width=0.75pt}} 

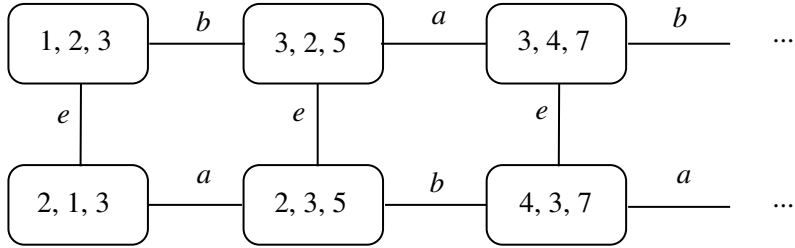
\begin{tikzpicture}[x=0.75pt,y=0.75pt,yscale=-1,xscale=1]
%

\draw   (158,77) .. controls (158,72.58) and (161.58,69) .. (166,69) -- (220,69) .. controls (224.42,69) and (228,72.58) .. (228,77) -- (228,101) .. controls (228,105.42) and (224.42,109) .. (220,109) -- (166,109) .. controls (161.58,109) and (158,105.42) .. (158,101) -- cycle ;
\draw   (158,158) .. controls (158,153.58) and (161.58,150) .. (166,150) -- (220,150) .. controls (224.42,150) and (228,153.58) .. (228,158) -- (228,182) .. controls (228,186.42) and (224.42,190) .. (220,190) -- (166,190) .. controls (161.58,190) and (158,186.42) .. (158,182) -- cycle ;
\draw   (276,77) .. controls (276,72.58) and (279.58,69) .. (284,69) -- (338,69) .. controls (342.42,69) and (346,72.58) .. (346,77) -- (346,101) .. controls (346,105.42) and (342.42,109) .. (338,109) -- (284,109) .. controls (279.58,109) and (276,105.42) .. (276,101) -- cycle ;
\draw   (276,158) .. controls (276,153.58) and (279.58,150) .. (284,150) -- (338,150) .. controls (342.42,150) and (346,153.58) .. (346,158) -- (346,182) .. controls (346,186.42) and (342.42,190) .. (338,190) -- (284,190) .. controls (279.58,190) and (276,186.42) .. (276,182) -- cycle ;
\draw   (398,77) .. controls (398,72.58) and (401.58,69) .. (406,69) -- (460,69) .. controls (464.42,69) and (468,72.58) .. (468,77) -- (468,101) .. controls (468,105.42) and (464.42,109) .. (460,109) -- (406,109) .. controls (401.58,109) and (398,105.42) .. (398,101) -- cycle ;
\draw   (398,158) .. controls (398,153.58) and (401.58,150) .. (406,150) -- (460,150) .. controls (464.42,150) and (468,153.58) .. (468,158) -- (468,182) .. controls (468,186.42) and (464.42,190) .. (460,190) -- (406,190) .. controls (401.58,190) and (398,186.42) .. (398,182) -- cycle ;
\draw    (228.7,89) -- (276.7,89) ;
\draw    (345,89) -- (398.2,89) ;
\draw    (228,170) -- (275.2,170) ;
\draw    (347,170) -- (398.2,170) ;
\draw    (194.4,109.2) -- (194.2,150.6) ;
\draw    (313.2,109.2) -- (313.2,150.6) ;
\draw    (434.4,109.2) -- (434.2,150.6) ;
\draw    (469,89) -- (520.2,89) ;
\draw    (469,171) -- (520.2,171) ;

\draw (171,162) node [anchor=north west][inner sep=0.75pt]   [align=left] {2, 1, 3};
\draw (290,162) node [anchor=north west][inner sep=0.75pt]   [align=left] {2, 3, 5};
\draw (413,162) node [anchor=north west][inner sep=0.75pt]   [align=left] {4, 3, 7};
\draw (172,81) node [anchor=north west][inner sep=0.75pt]   [align=left] {1, 2, 3};
\draw (290,82) node [anchor=north west][inner sep=0.75pt]   [align=left] {3, 2, 5};
\draw (412,81) node [anchor=north west][inner sep=0.75pt]   [align=left] {3, 4, 7};
\draw (241,71) node [anchor=north west][inner sep=0.75pt]   [align=left] {\,\,\,\,$b$};
\draw (369,71) node [anchor=north west][inner sep=0.75pt]   [align=left] {$a$};
\draw (490,69) node [anchor=north west][inner sep=0.75pt]   [align=left] {$b$};
\draw (251,151) node [anchor=north west][inner sep=0.75pt]   [align=left] {$a$};
\draw (492,151) node [anchor=north west][inner sep=0.75pt]   [align=left] {$a$};
\draw (368,152) node [anchor=north west][inner sep=0.75pt]   [align=left] {$b$};
\draw (181,121) node [anchor=north west][inner sep=0.75pt]   [align=left] {$e$};
\draw (299,120) node [anchor=north west][inner sep=0.75pt]   [align=left] {$e$};
\draw (421,119) node [anchor=north west][inner sep=0.75pt]   [align=left] {$e$};
\draw (540,85) node [anchor=north west][inner sep=0.75pt]   [align=left] {...};
\draw (540,168) node [anchor=north west][inner sep=0.75pt]   [align=left] {...};
\end{tikzpicture}
}

\smallskip\par\noindent
If \emph{asked the same question again}, they again answer ``\textit{I don't know}'', which deletes $(2,1,3)$ and $(1,2,3)$:

\medskip\par\noindent

{\tikzset{every picture/.style={line width=0.75pt}}        

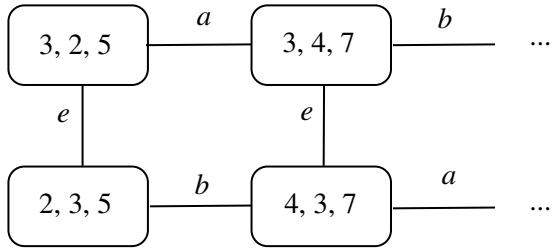
\begin{tikzpicture}[x=0.75pt,y=0.75pt,yscale=-1,xscale=1]
%

\draw   (210,73) .. controls (210,68.58) and (213.58,65) .. (218,65) -- (272,65) .. controls (276.42,65) and (280,68.58) .. (280,73) -- (280,97) .. controls (280,101.42) and (276.42,105) .. (272,105) -- (218,105) .. controls (213.58,105) and (210,101.42) .. (210,97) -- cycle ;
\draw   (210,154) .. controls (210,149.58) and (213.58,146) .. (218,146) -- (272,146) .. controls (276.42,146) and (280,149.58) .. (280,154) -- (280,178) .. controls (280,182.42) and (276.42,186) .. (272,186) -- (218,186) .. controls (213.58,186) and (210,182.42) .. (210,178) -- cycle ;
\draw   (332,73) .. controls (332,68.58) and (335.58,65) .. (340,65) -- (394,65) .. controls (398.42,65) and (402,68.58) .. (402,73) -- (402,97) .. controls (402,101.42) and (398.42,105) .. (394,105) -- (340,105) .. controls (335.58,105) and (332,101.42) .. (332,97) -- cycle ;
\draw   (332,154) .. controls (332,149.58) and (335.58,146) .. (340,146) -- (394,146) .. controls (398.42,146) and (402,149.58) .. (402,154) -- (402,178) .. controls (402,182.42) and (398.42,186) .. (394,186) -- (340,186) .. controls (335.58,186) and (332,182.42) .. (332,178) -- cycle ;
\draw    (279,85) -- (332.2,84.6) ;
\draw    (281,166) -- (332.2,165.6) ;
\draw    (247.2,105.2) -- (247.2,145.6) ;
\draw    (368.4,105.2) -- (368.2,145.6) ;
\draw    (403,85) -- (454.2,84.6) ;
\draw    (403,167) -- (454.2,166.6) ;

\draw (224,158) node [anchor=north west][inner sep=0.75pt]   [align=left] {2, 3, 5};
\draw (347,158) node [anchor=north west][inner sep=0.75pt]   [align=left] {4, 3, 7};
\draw (224,78) node [anchor=north west][inner sep=0.75pt]   [align=left] {3, 2, 5};
\draw (346,77) node [anchor=north west][inner sep=0.75pt]   [align=left] {3, 4, 7};
\draw (303,67) node [anchor=north west][inner sep=0.75pt]   [align=left] {$a$};
\draw (424,65) node [anchor=north west][inner sep=0.75pt]   [align=left] {$b$};
\draw (426,147) node [anchor=north west][inner sep=0.75pt]   [align=left] {$a$};
\draw (302,148) node [anchor=north west][inner sep=0.75pt]   [align=left] {$b$};
\draw (233,116) node [anchor=north west][inner sep=0.75pt]   [align=left] {$e$};
\draw (355,115) node [anchor=north west][inner sep=0.75pt]   [align=left] {$e$};
\draw (470,82) node [anchor=north west][inner sep=0.75pt]   [align=left] {...};
\draw (470,165) node [anchor=north west][inner sep=0.75pt]   [align=left] {...};
\end{tikzpicture}
}

\smallskip\par\noindent
If \textit{asked for the 3rd time}, Alice says `` \textit{I know Bob's number}'' and Bob says ``\textit{I don't know}''. The announcement $! (K_a x(b) \wedge \neg K_b x(a))$ \emph{deletes all states but $(2,3,5)$}: so we have $x(a)=2, x(b)=3, x(e)=5$. 

\medskip

\par\noindent\textbf{Example 3} 
Starting again in the initial situation from Example 1, agents are publicly announced that $x(a)<x(b)$. After that, Alice knows Bob's number, i.e. $[!(x(a)<x(b))] K_a x(b)$. We \emph{can} reason hypothetically about this scenario \textit{in the original model, without updating it}, using \textit{conditional knowledge $K_a^{x(a)< x(b)} x(b)$ of the (hypothetical) value of $x(b)$ given the condition $x_a<x_b$}. We can also name this hypothetical value using a conditionalized version of the ``least-of-$N$" description introduced above, e.g. writing $1_1.x(b)_a^{x(a)<x(b)}$ for the unique (=least out of only one) value of $x(b)$ that is considered possible by Alice conditional on $x(a)<x(b)$.  But note that, when considering a condition that is \emph{known to be false} (e.g., $x(a)=x(b)$), there are \emph{no such hypothetical values}: so the best we can do is to interpret the `least' value as the `infimum' $1_1. x(b)_a^{x(a)=x(b)}= inf\,\emptyset= \infty$. This means that we are working in $\mathbb{N}^\top:=\mathbb{N}\cup \{\infty\}$.

\medskip

\par\noindent\textbf{Example 4: Back in time}
Imagine now that our story starts \textbf{earlier}, at a time when each `agent' (Alice $a$, Bob $b$, and the envelope $e$) only `knows' their own number (but not yet the correlations between numbers).
The model is just $\mathbb{N}\times \mathbb{N}\times \mathbb{N}$, with $(x(a), x(b), x(c))\sima (x'(a), x'(b), x'(c))$ iff $x(a) =x'(a)$,
and similarly for $b$ and $e$. Next, the correlations are publicly announced $!((x(a)=x(b)+1\vee x(b)=x(a)+1)\wedge x(e)=x(a)+x(b))$, and after that the semi-public hacking !(a:b) happens. Can we check that Alice will know $x(e)$ after this scenario \emph{without updating the model }$\mathbb{N}\times \mathbb{N}\times \mathbb{N}$?
For this, we will need an operator $K_{a,b}^{  (x(a)=x(b)+1\vee x(b)=x(a)+1) \wedge (x(e)=x(a)+x(b))} x(e)$ for \emph{conditional distributed knowledge of a value}.

\medskip

To formalize the examples, we need the following ingredients:
\emph{equality} $x=y$; \emph{order} $x<y$;
\emph{functions} $x=y+1$, $z=x+y$;
\emph{(conditional) (distributed) knowledge of a number} $K_A^\varphi x$;
a way to 
express than an agent/group can \emph{(conditionally) narrow down to $N$ possible values} ($|x|_A^\varphi \leq N$),
and to \emph{name these values in increasing order}: $1_N.x_A^\varphi$, $2_N.x_A^\varphi$, etc;
\emph{public announcements} $!\varphi$, \emph{semi-public hacking} $!(a:b)$,  etc. 

\vspace{-4mm}

\section{Syntax and Semantics of $LDA$ and $LDAE$}\label{Logic}
\vspace{-2mm}

\par\noindent\textbf{Vocabularies: static and dynamic.} 
A \emph{static vocabulary} $\mathcal{V}=(\Agents, V, Prop, 
Pred, Funct, ar, \epsilon)$ consists of the following:\footnote{All the sets ($\Agents, V, Prop, 
Pred, Funct,\mathcal{E})$ in a vocabulary are assumed to be mutually disjoint.}
a finite set $\Agents$ of \emph{agents} $a, b, c, \ldots$, also called `locations' (e.g., websites, databases, processors, etc), where data are stored and processed;
a set $V$ of \emph{basic variables} $v$, $v', v'', \ldots$;
a set $Prop$ of \emph{atomic propositions} $p, q, \ldots$;
a set $Pred$ of \emph{predicate symbols} $P, Q, \ldots$, 
including 
\emph{equality} ($=$) and an \emph{order predicate} $\leq$; 
a set $Funct$ of \emph{function symbols} $F, G, \ldots$, including two \emph{constants} $\bot, \top$; 
an \emph{arity} map $ar: Pred\cup Funct\to \mathbb{N}$, sending predicate symbols $P\in Pred$ and function symbols $F\in Funct$ to natural numbers $ar(P), ar(F)\in \mathbb{N}$,  with $ar(\top)=ar(\bot)=0$ and $ar(=)=ar(\leq)=2$;
an auxiliary symbol $\epsilon\not\in \Agents\cup V\cup Prop\cup Pred\cup Funct$, denoting ``the \emph{environment}'' (e.g., the `envelope' in our example).

\smallskip 
\par A \emph{dynamic vocabulary} $(\mathcal{V}, \mathcal{E})$ is a pair consisting of a static vocabulary $\mathcal{V}$ and a \emph{countable} set $\mathcal{E}$ of \emph{`event names'} (denoted by $e, e', f, \ldots$), that includes some special symbols $!, ?, \tau$.

\smallskip 

\par\noindent\textbf{Groups} Given the static vocabulary $\mathcal{V}$, a \emph{group} $A\subseteq\Agents$ is any non-empty set of agents in $\Agents$. We use capital letters $A, B, \ldots$ to denote groups. 

\medskip

\par\noindent\textbf{Ordered Value Domains} A \emph{value domain} (for a given vocabulary $\mathcal{V}$) is a first-order model $\bD=(D, I)$ for $\mathcal{V}$, consisting of: a \emph{domain of `values'} $D$, of cardinality $|D|>1$; and an \emph{interpretation function} $I$, mapping each functional symbol $F$ of arity $n$ into a function $I(F)=F_\bD: D^n\to D$ (so that the constants $\bot$ and $\top$ are interpreted as values $\bot_\bD, \top_\bD\in D$), and each
relational symbol $P$ of arity $n$ into a set $I(P)=P_\bD\subseteq D^n$ of $n$-tuples of values in $D$, subject to the following conditions:
the equality symbol $=$ is interpreted as the \emph{identity relation} $=_\bD$ on $D$; and the order symbol $\leq$ is interpreted as \emph{some total order} $\leq_\bD\subseteq D\times D$ on the set $D$ having \emph{$\bot_\bD$ as its bottom element and $\top_\bD$ as its top element} (i.e., $\bot_\bD \leq_\bD\, d \leq_\bD \top_\bD$ for all $d\in D$).

\smallskip
\par\noindent\textbf{Minimization} Given an ordered value domain $\bD=(D, I)$, it is easy to see that \emph{every finite non-empty subset $D'\subseteq D$} \emph{has a unique minimum} $min\, D'\in D'$, s.t. $min\, D'\leq_\bD d'$ for all $d'\in D'$. 

\smallskip

\par\noindent\textbf{Notation: `Cutoff Minimization'} We can also consider, for any natural number $N\geq 1$, an ``$N$-cutoff'' version of minimization $min_N$, defined on \emph{all} sets $D'\subseteq D$ (including the infinite ones), by putting:
\vspace{-2mm}
\[
\begin{array}{lllll}
min_N \, D' \  \ & := \ & \top_\bD \,  (\mbox{`trivial'}), & \mbox{ when $D'=\emptyset$,}
\ \\
min_N \, D' \  \ & := \ \ & min\, D', & \mbox{ when  $1\leq |D'|\leq N$,} \,\, \mbox{ and}
\ \\
min_N \, D' \  \ & := \ \ & \bot_\bD \,  (\mbox{`undefined'}),  & \mbox{ otherwise (for $|D'|> N$).}
\vspace{-2mm}
\end{array}
\]
\par\noindent\textbf{Notation: `Cutoff' $n^{th}$ element} 
For all sets $D'\subseteq D$ and all $1\leq n\leq N$, we can now introduce a notation $n_N (D')$ for \emph{``the $n^{th}$ element of the $\leq N$ elements'' of $D$}, by putting:
\vspace{-2mm}
$$1_N (D')\, \, := \, \, min_N\, D', \,\,\,\,\,\,\,\,\,\,\,\,\,
(n+1)_N (D') \,\, :=\,\, min_{N-1}\, (D'- \{n_N (D')\})  \, \,\,\, \mbox{ for $1\leq n < N$}.
\vspace{-2mm}$$

\par\noindent\textbf{Epistemic State Models} 
An \emph{epistemic state model over a value domain $\bD=(D, I)$} is a tuple $\bM= (S, \sim, \underline{\bullet}(\bullet))$, where:
(i) $S$ is a set of \emph{states} (or `possible worlds'), typically denoted by $s, w, \ldots$; 
(ii) $\sim: \Agents\to \mathcal{P}(S\times S)$ maps agents $a\in\Agents$ to \emph{equivalence relations} $\sima\subseteq S\times S$,
called \emph{`indistinguishability' relations};
(iii) $\underline{\bullet}(\bullet):S\times (V\cup Prop) \to D$ is an \emph{assignment map} (`valuation'), mapping pairs $(s,v)\in S\times V$ into
arbitrary values $\underline{s}(v)\in D$, and mapping pairs $(s,p)\in S\times Prop$ into values $\us(p)\in\{\bot_\bD, \top_\bD\}$.
\smallskip
\par\noindent\textbf{Group indistinguishability} Given a state model $\bM= (S, \sim, \underline{\bullet}(\bullet))$, we define \emph{group indistinguishability relations} $\simA :=\bigcap_{a\in A} \sima$ on $S$ (for all groups $A\subseteq \Agents$) by \emph{taking intersections}.

\medskip
\par\noindent\textbf{Syntax of $LDAE$} For a fixed dynamic vocabulary $(\mathcal{V}, \mathcal{E})$, the \emph{(dynamic) Logic of (group) Data Access \& Exchange} ($LDAE$) has a syntax consisting of: (1) a set $Var:=Var (\mathcal{V})$ of compound \emph{variables} (or
`terms') $x$; (2) a set $Fml:=Fml(\mathcal{V})$ of \textit{formulas} $\varphi$;  (3) a set of \emph{`syntactic' event models}; (4) a set $Events$ of \emph{data-exchange events}.
These components are simultaneously defined by mutual recursion, as follows:
\\
\par (1) \& (2) \textbf{Variables and Formulas} 
The sets $Var$ and $Fml$ are given by the clauses:
\vspace{-2.5mm}
$$
\begin{array}{ccc ccc cc cc cc}
x & ::= &  
v &|&
F(\overline{x}) &|&\varphi\!\!\to\!\! x|x &|& \mu_N\, x_A^\varphi
&|&
\be(x)
\\
\varphi & :: = &  
p  &|&  P\overline{x}  &|& \varphi \to \varphi         &|&
K_A \varphi 
&|&
[\be] \varphi
\end{array}\vspace{-2.5mm}
$$ 
where: $v\in V$ are basic variables; $p\in Prop$ are atoms;
$\overline{x}=(x_1, \ldots, x_n)$ are tuples of terms; $P\in Pred$ are $n$-ary predicate symbols; $F$ are $n$-ary function symbols; $a\in \Agents$ are agents; $A\subseteq \Agents$ are groups; $N\geq 1$ are integers; and the \emph{events} $\be$  are technically ``pointed event models'' $(\bE,e)$, as defined below.
\\
\par (3) \& (4) \textbf{Event Models and Events} An \emph{event model $\bE=(E, \sim, \underline{\bullet}(\bullet))$} consists of:
(i) a finite set $E\subseteq \mathcal{E}$ of event names;
(ii) an equivalence relation $\sima\subseteq E\times E$ (\emph{agent $a$'s indistinguishability} over events)
for each $a\in\Agents$;
(iii) a \emph{change map} $\underline{e}(\bullet): Prop\cup V\cup \Agents\cup\{\epsilon\}\to  Fml\cup Var\cup \mathcal{P}(\Agents)$ for each event $e\in E$, 
mapping $\epsilon$ to a formula $\ue(\epsilon)\in Fml$, atoms $p\in Prop$ to formulas $\ue(p)\in Fml$, basic variables $v\in V$ to terms $\underline{e}(v)\in Var$, and agents $a\in \Agents$ to groups $\underline{e}(a)\subseteq \Agents$.
We require these items to satisfy \textit{two conditions}:
    \begin{enumerate}
\item \emph{Self-Access} (agents access their own database): $a\in \underline{e}(a)$;
\item \emph{Known Access} (agents know their sources): $e \sima f$ \mbox{ implies } 
$\underline{e}(a)=\underline{f}(a)$;
\end{enumerate}
As mentioned, a \emph{data-exchange event} is just a \emph{`pointed' event model} $\be=(\bE, e)$, i.e. a \emph{pair} of an event model $\bE=(E, \sim, \underline{\bullet}(\bullet))$ and an event name $e\in E$. We denote by $Events$ the set of data-exchange events.

\medskip

\par\noindent\textbf{The Static Logic $LDA$} The `static' fragment of our logic, called the \emph{(static) Logic of (group) Data Access} ($LDA$), consists of all formulas of $LDAE$ that are built without the use of dynamic operators $[\be]\varphi$ or $\be(x)$. 

\smallskip  

\par\noindent\textbf{Precondition and Postconditions} 
The \emph{precondition} of any event $\be=(\bE,e)$
is the formula $pre_\be:=\underline{e}(\epsilon)$, which intuitively gives \emph{the event's condition of possibility}: $\be$ can only happen in states satisfying its precondition $pre_\be$. The \emph{event $\be$'s postcondition for $p\in Prop$} is the formula $post_\be (p):=\ue(p)$; similarly,
the \emph{event's postcondition for $v\in V$} is the term $post_\be (v):= \ue(v)$. Intuitively, the postconditions determine the way an event changes the values of atoms and variables: the new value of variable $v$ (or atom $p$) after event $\be$ coincides with the `old' value of the formula $post_\be (p)$ (or the term $post_\be (v)$) before the event.

\smallskip  

\par\noindent\textbf{(Extended) Access Map}
\emph{Event $\be$'s access map} is the restriction of $\ue$ to $\Agents$, specifying for each agent $a\in\Agents$ the group 
$\ue(a)\subseteq \Agents$ of \emph{all sources/agents (whose locations/databases are) accessed by agent $a$ during the event}. We can also \emph{extend} the access map to \emph{groups} $A\subseteq \Agents$, by putting $\underline{e}(A):=\bigcup \{\underline{e}(a): a\in A\}$, for the set of \emph{sources that are distributedly accessible to the group $A$ during event $e$}.
\smallskip

\par\noindent\textbf{Subexpression-complexity} An \emph{expression} $\alpha$ is any formula $\varphi$, term $x$ or event $\be$ of $LDAE$. The \emph{sub-expression complexity order} $<$ is the least transitive relation s.t.: (1) every formula is $>$ its subformulas, and also $>$ all terms and events occurring in it; (2) every term is $>$ its subterms, and also $>$ all formulas and events in it; (3) every event $\be=(\bE,e)$ is $>$ all preconditions and postconditions of the form $pre_f$, $\underline{f}(p)$ and $\underline{f}(v)$ with $f\in E$, $p\in Prop$ and $v\in V$.
It is easy to see that \emph{$<$ is a well-founded partial order}.

\smallskip
\par\noindent\textbf{Semantics} We simultaneously define three  semantic notions 
on state models $\bM= (S, \sim, \underline{\bullet}(\bullet))$: 
\vspace{-1mm}
\begin{enumerate}
\item a \emph{satisfaction relation} $s\models_{\bM} \varphi$ between states $s\in S$ (in any state model $\bM$) and formulas $\varphi\in Fml$;
\item an \emph{extended assignment/valuation function} from $S\times (Var\cup Fml)$ to $D$, mapping $(state, variable)$-pairs $(s,x)\in S\times Var$
into arbitrary values $s(x)_\bM\in D$, and mapping  $(state, formula)$-pairs $(s,\varphi)\in S\times Fml$
into extreme values $s(\varphi)_\bM\in \{\bot_\bD, \top_\bD\}$. 
\item a \emph{product update operation}, mapping state models $\bM=(S,\sim,\underline{\bullet}(\bullet))$ and event models $\bE=(E,\sim,\underline{\bullet}(\bullet))$ into \emph{updated state models} $\bM\bigotimes \bE=(S\otimes E,\sim, \underline{\bullet}(\bullet))$ (over the same value domain $\bD$).
\end{enumerate}
\par\noindent For this definition, we need an auxiliary notation: for variables $x\in Var$, groups $A$, formulas $\varphi$ and states $s\in S$, the \emph{set of possible values of $x$ at state $s$ according to group $A$ conditional on $\varphi$} is
$$(x_A^\varphi)_{s,\bM} \, :=\, \{w(x)_\bM : w\simA s, w\models_\bM\varphi\}.
$$
With this notation, we define our three semantic notions by mutual recursion:
\[
\begin{array}{lllll}
s\models_\bM p \  \ & \mbox{iff} \ \ && \underline{s}(p)=\top_\bD\\
s\models_\bM P x_1\ldots x_n \  \ & \mbox{iff} \ \ && (s(x_1)_\bM, \ldots, s(x_n)_\bM)\in I(P)\\
s\models_\bM\varphi\to\psi \  \ & \mbox{iff} \ \ &&  s\models_\bM\varphi \mbox{ implies } s\models_\bM\psi
\\
s\models_\bM K_A \varphi \ \ & \mbox{iff} \ \ && w\models_\bM\varphi \mbox{ for all $w\simA s$}
\\
s\models_\bM [\be] \varphi \ \ & \mbox{iff} \ \ && (s,e)\in S\otimes E \, \mbox{ implies } (s,e)\models_{\bM\bigotimes \bE} \varphi,
\,\,\,\,\,
\mbox{ where } \be=(\bE,e).
\vspace{2mm}
\\
s(v)_\bM \ \ & = \ \ && \underline{s}(v) \,\,\,\,\,\,\,\,\,\, \,\,\,\,\, \,\,\,\,\, \mbox{ as given in the model $\bM$} \\
s(\varphi\!\!\to\!\! x|y)_\bM  \ \ & = \ \ && s(x)_\bM  \,\,\,\,\, \,\,\,\,\, \,\,\,\,\, \mbox{ if $s\models_\bM\varphi$, and }\\
s(\varphi\!\!\to\!\! x|y)_\bM  \ \ & = \ \ && s(y)_\bM  \,\,\,\,\, \,\,\,\,\, \,\,\,\,\, \mbox{ if $s\not\models_\bM\varphi$ }\\
s(F(x_1,\ldots, x_n))_\bM  \ \ & = \ \ && (I(F)) (s(x_1)_\bM, \ldots, s(x_n)_\bM)
\\
s(\mu_N x_A^\varphi)_\bM \ \ & = \ \ &&  min_N \, Val(x_A^\varphi)_{s,\bM} \\
s(\be(x))_\bM \ \ &  = \ \ && (s,e)(x)_{\bM\bigotimes \bE} \,\,\,\,\, \,\,\,\,\, 
\mbox{ if $\be=(\bE,e)$ is s.t. $(s,e)\in S\otimes E$, and}\\
s(\be(x))_\bM \ \ &  = \ \ && \top_\bD \,\,\,\,\,  \,\,\,\,\,  \,\,\,\,\,\,\,\,\,\,\,\,\,\,\, \,\,\,\,\,\,\,\,\,\, \,\,\,\,\,\,\,\,\,\,\mbox { otherwise.}
\\
s(\varphi)_\bM \ \ & = \ \ && \top_\bD  \,\,\,\,\, \,\,\,\,\, \,\,\,\,\, \mbox{ if $s\models_\bM\varphi$, and }\\
s(\varphi)_\bM \ \ & = \ \ && \bot_\bD  \,\,\,\,\, \,\,\,\,\, \,\,\,\,\, \mbox{ if $s\not\models_\bM\varphi$.}
\vspace{2mm}\\
\bM=(S, \sim, \underline{\bullet}(\bullet)), \,\, \bE=(E, \sim, \underline{\bullet}(\bullet)) \ \ &  \mapsto \ \ && \bM\bigotimes\bE=(S\otimes E, \sim, \underline{\bullet}(\bullet)), \,\,\, \,\,\,\,\,\,\, \mbox{  where: } 
\\\
S\otimes E \ \ & = \ \ && \{(s,e)\in S\times E \mid s\models_\bM pre_\be\}
\\\
(s,e)\sim_a (s', e')   \ \ & \mbox{iff} \ \ && s\sim_{\underline{e}(a)} s' \mbox{ and } e\sim_a e'
\\\
\underline{(s,e)} (p)  \ \ & =   \ \ &&  s (post_\be(p))_\bM
\\\
\underline{(s,e)}(v)  \ \ & = \ \ && s(post_\be (v))_\bM
\end{array}
\]

\medskip
\par\noindent\textbf{Extended assignment on sets of expressions} We can `lift' the assignment map to the level of \emph{sets} of variables $X\subseteq Var$ and \emph{sets} of formulas $\Phi\subseteq Fml$), by putting:
$s(X)_\bM \, :=\, \{s(x)_\bM : x\in X\}$; $s(\Phi)_\bM\, :=\, \{s(\varphi)_\bM : \varphi\in \Phi\}$.
Whenever the model is understood, we skip the subscript $\bM$, writing simply $s\models \varphi$, $s(x)$, $s(X)$ and $s(\Phi)$. For sets of formulas $\Phi$, we also write $s\models \Phi$ whenever $s\models\varphi$ for all $\varphi\in \Phi$. 
\smallskip

\par
\noindent\textbf{Abbreviations}. We define the usual Boolean connectives $true := (\top=\top)$, $false:= (\top=\bot)$,  
$\neg\varphi := (\varphi\to false)$, $\varphi\wedge \psi$, $\varphi\vee\psi$, $\varphi\leftrightarrow \psi$, as well as the dual existential (Diamond) modalities $\langle K_A\rangle \varphi  :=\neg K_A\neg \varphi$, $\langle K_A^\theta\rangle \varphi := \neg K_A^\theta\neg \varphi$, $\langle \be\rangle \varphi   := \neg [\be] \neg \varphi$. 
We also use the following abbreviations, for variables $x,y\in Var$, finite sets of variables $X,Y\subseteq Var$ and natural numbers $n,N$ with $1\leq n\leq N$:
\vspace{1mm}

\centerline{$x\in Y \, :=\, \bigvee\{ x=y: y\in Y\},  \,\,\, \,\,\,\,\,\, ?_\varphi \, :=\, \varphi\!\!\to\!\! \top|\bot, \,\,\, \,\,\,\,\,\,
\varphi\!\!\to\!\! x \, :=\, \varphi\!\!\to\!\! x|\top,$}
\vspace{1.5mm}
\centerline{$K_A^\theta \varphi \, :=\, K_A (\theta\to \varphi),  \,\,\, K_A^\theta x  \, :=\, K_A^\theta (x= \mu_1 x_A^\theta), \,\,\,\,\,
K_A x \,\, :=\,\, K_A^{true} x,  
$}
\vspace{1.5mm}
\centerline{$1_N.x_A^\theta \,\, :=\,\, \mu_N x_A^\theta,  \,\, \,\,\,
(n+1)_N.x_A^\theta \,\, :=\,\, \mu_{N-1} x_A^{\theta \wedge x> n_N.x_A^\theta}, $}
\vspace{1.5mm}
\centerline{$  |x|_A^\theta\leq 0 \,\, :=\,\, K_A \neg\theta,
 \,\,\,\,\, \,\,\,\,\, |x|_A^\theta \leq n \,\, :=\,\, K_A^\theta (x\in Var^n (x_A^\theta)), \,\,\,\, \mbox { where } Var^n(x_A^\theta):=
\{i_n.x_A^\theta : 1\leq i\leq n\},$}
\vspace{1.5mm}
\centerline{$|x|_A^\theta= 0 \,\, :=\,\, |x|_A^\theta\leq 0, \,\,\,\,\,
 |x|_A^\theta > n \,\, :=\,\, |x|_A^\theta\not\leq n, $}
 \vspace{1.5mm}
 \centerline{$
 |x|_A^\theta = (n+1) \,\, :=\,\, |x|_A^\theta> n \wedge |x|_A^\theta \leq (n+1)\, \mbox{ for $n\geq 0$},$}
 \vspace{1.5mm}
\centerline{$min \, \{x\} \, :=\,  x,  \,\,\,  min \, (X\cup \{y\})  \, :=\, (min \, X\leq y)\!\!\to\!\! (min\, X)|y.$}
\noindent Intuitively, $x\in Y$ says that $x$ currently takes the same value as some term in $Y$. The term $?_\varphi$ is a Boolean variable, taking value $\top_\bD$ if $\varphi$ is true, and value $\bot_\bD$ otherwise. The operator $\varphi\!\!\to\!\! x$ is the term analogue of material implication: it takes the value of $x$ if $\varphi$ is true, and value $\top_\bD$ otherwise. 
Next, \emph{`knowledge of value' is definable in our logic}, in both conditional and unconditional forms, for groups and individuals, via the abbreviations $K_A^\theta x$, $K_A x$, $K_A^\theta\ux$, $K_A\ux$. Note that, for the empty tuple $\lambda=()$, we have $K_A^\theta \lambda = \bigwedge \emptyset =\top$. The term $n.x_A^\theta$ denotes the $n^{th}$ value of $x$ (in $\leq_\bD$-order) considered possible by $A$, conditional on $\theta$. The expressions $|x|_A^\theta =n$, $|x|_A^\theta \leq n$ and $|x|_A^\theta >n$ refer to the cardinality of the set of possible values of $x$, according to group $A$, given $\theta$. Finally, $min\, X$ denotes the minimum value of all terms in $X$.

\medskip

\par\noindent \textbf{Distributed Location} For any \emph{expression}, i.e., any term, formula or event $\alpha\in Var\cup Fml\cup Events$, its \emph{(distributed) location} $\mathcal{L}(\alpha)\subseteq\Agents \cup \{\epsilon\}$ is
given by the following recursive clauses:
\vspace{1mm}

\centerline{$ \mathcal{L}(p)= \mathcal{L} (v) \, \, =\,\, \{\epsilon\}, \,\,\, \,\,\,\,  \, \,\,\,\,\, \, \,\,\,\,\,\, \,\,\,\,\,  \,\,\, \,\,\,\, \mathcal{L}(K_A \varphi)=
\mathcal{L}(\mu_N x_A^\varphi)\,\, =\,\, A,$}
\vspace{1mm}
\centerline{$\,\,\,\,\, \mathcal{L}(Px_1\ldots x_n)=\mathcal{L}(F(x_1, \ldots, x_n)) \,\, =\,\, \bigcup \{\mathcal{L}(x_i):1\leq i\leq n\},$}
\vspace{1mm}
\centerline{$\mathcal{L}(\varphi\to \psi)  \,\, =\,\, \mathcal{L}(\varphi)\cup \mathcal{L}(\psi),  \,\,\, \,\,\,\, \, \,\,\,\,\, \,\,\,\, \, \,\,\,\,\, \, \,\,\,\,\,\mathcal{L}(\varphi\!\!\to\!\! x|y)\,\, =\,\, \mathcal{L}(x)\cup \mathcal{L}(\varphi)\cup \mathcal{L}(y),
$}
\vspace{1mm}
\centerline{$
\mathcal{L}(\be)\,\, = \,\, \mathcal{L}(pre_\be), \,\,\, \,\,\,\, \, \,\,\,\,\, \, \,\,\,\,\,
\mathcal{L}([\be] \varphi) \,\, = \,\, \mathcal{L}(\be)\cup \underline{e}(\mathcal{L}(\varphi)),   \,\,\,\, \, \,\,\,\,\, \, \,\,\,\,\,
\mathcal{L}(\be (x)) \,\, = \,\, \mathcal{L}(\be)\cup \underline{e}(\mathcal{L}(x)),$}
\vspace{1mm}
\noindent where we used the extended access map $\underline{e}(A)$.
In particular, this gives us that $\mathcal{L}(\bot)=\mathcal{L}(\top)=\mathcal{L}(true)=\mathcal{L}(false)=\emptyset$, $\mathcal{L}(\neg\varphi)=\mathcal{L}(\varphi)$, $\mathcal{L}(\varphi\wedge \psi)= \mathcal{L}(\varphi)\cup \mathcal{L}(\psi)$ and $\mathcal{L}(K_A^\theta x)=A$. 

We can also extend the location map to \emph{finite sets of terms} $X\subseteq Var$, by putting:
$$\mathcal{L}(X) \,\, :=\,\, \bigcup_{x\in X} \mathcal{L}(x).$$
\par
\noindent\textbf{Locality} The values of terms or propositions having a distributed location within a group $A\subseteq \Agents$ are \emph{distributed knowledge among group $A$'s members}, as shown by the results below.
\begin{proposition}(\emph{Preservation})\label{Preservation} 
Let $s,w$ be states in a model $\bM$ s.t. $s\simA w$ for some group $A\subseteq\Agents$, let $\varphi$ be a static formula s.t. $\mathcal{L}(\varphi)\subseteq A$, and let $x\in Var$ be a static term s.t. $\mathcal{L}(x)\subseteq A$. Then we have:
\begin{enumerate}
\item $s\models_\bM \varphi$ iff $w\models_\bM \varphi$;
\item $s(x)_\bM=w(x)_\bM$.
\end{enumerate}
\end{proposition}

\begin{corollary}\label{Local Knowledge} (\emph{Local Knowledge})
For formulas $\varphi$ and terms $x$, we have the following validities:
$$\models\, \varphi \to K_A \varphi, \,\,\, \mbox{ whenever  $\mathcal{L}(\varphi)\subseteq A$}; 
\,\,\,\,\,\,\,\,\,\,\,\,\, \,\,\,\,\,\,\,\,\,\,\,\,\, \,\,\,\,\,\,\,\,\,\,\,\,\, \,\,\,\,\,\,\,\,\,\,\,\,\, \models\, K_A x,  \,\,\, \mbox{ whenever  $\mathcal{L}(x)\subseteq A$.}$$
\end{corollary} 

\subsection{Examples of Data-Exchange Events and Event Models}
\vspace{-2mm}

In \emph{`semi-public' event models} (with \emph{only one event}), it is common knowledge who can read whose data.

\smallskip

\par\noindent\textbf{Public Announcements}
For every formula $\varphi$, the event \emph{$!\varphi= (\bE_{!\varphi}, !)$ of publicly announcing $\varphi$} has an event model $\bE_{!\varphi}= (\{!\}, \sim, \underline{\bullet}(\bullet))$ whose only event name is the special symbol $!$, with identity $\sim_a =\{(!,!)\}$ as accessibility relation for every agent, and where we put $\underline{!}(p)=p$, $\underline{!}(v)=v$, $\underline{!}(a)=\{a\}$ and $\underline{!} (\epsilon)=\varphi$, hence $pre_{!\varphi}=\varphi$. This event encodes the dynamics of truthful public announcements \cite{Plaza}: the updated model $\bM \otimes \bE_{!\varphi}$ is (isomorphic to the one) obtained by deleting all the $\varphi$-worlds from the original model $\bM$ (and keeping everything else the same).

\smallskip

\par\noindent\textbf{Public and Semi-Public Sharing} (\emph{``Tell Us All You Know''}) For groups $A,B\subseteq\Agents$, $!(A:B)= (\bE_{!(A:B)}, !)$ is a semi-public event whose model $\bE_{!(A:B)}= (\{!\}, \sim, \underline{\bullet}(\bullet))$ has the same structure as  $\bE_{!\varphi}$, except for the access map and the precondition, which are given by putting: $\underline{!}(a)=B\cup\{a\}$ for all $a\in A$: $\underline{!}(a)=\{a\}$ for $a\not\in A$; and 
$\underline{!} (\epsilon)=true$ (so the precondition $pre_{!(A:B)}=true$ is tautological). In this action, it is common knowledge that \emph{all agents in $A$ gain access to the databases of all agents in $B$}. Its effect is to replace all the relations $\sima$ (with $a\in A$) from the original model $\bM$ by the new relations
$\sima^{!(A:B)} := \sima \cap \simB$ in the updated model $\bM\otimes \bE_{!(A:B)}$, while keeping everything else unchanged (including the relations $\sim_a$ with $a\not\in A$). A special case is semi-public sharing $!(a: b)$ from $b$ to $a$, obtained by taking $A=\{a\}$ and $B=\{b\}$: it is common that $b$ shares all his knowledge with $a$.
Another special case is \emph{$A$-public sharing} $!B:= !(\Agents: B)$, which is obtained by taking $A=\Agents$: all agents in $B$ publicly share all their information. An even more special case is \emph{$b$-public sharing} $!b=!\{b\}= !(\Agents:\{b\})$, obtained by taking $A=\{a\}$ for some designated agent $a$ (and $B=\Agents$): agent $a$ publicly ``tells all she knows''.\footnote{This corresponds to a dynamics 
introduced and axiomatized in this form in \cite{Baltag2010}, though technically isomorphic to the action of public resolution of an issue/question \cite{BenthemMinica1}.}

\smallskip

\par\noindent\textbf{Sharing \emph{between} Multiple Groups} For groups $A_1, B_1, \ldots, A_n, B_n$, the event $!(A_1:B_1, \ldots, A_n: B_n)$ is a semi-public event, whose model $\bE_{!(A_1:B_1, \ldots, A_n: B_n)}= (\{!\}, \sim, \underline{\bullet}(\bullet))$ has the same structure as $\bE_{!(A:B)}$, except that the access map is given by $\underline{!}(a)=\{a\}\cup \bigcup \{B_i: i\leq n, a\in A_i\}$ for all $a$: it is common knowledge that \emph{all agents in each group $A_i$  gain access to the databases of all agents in the corresponding group $B_i$}.

\smallskip

\par\noindent\textbf{Sharing \emph{within} Groups} For groups $A_1, \ldots, A_n\subseteq \Agents$, the event of \emph{parallel sharing within different groups} \linebreak $!(A_1, \ldots, A_n)=!(A_1:A_1, \ldots, A_n: A_n)$ is the special case of  $!(A_1:B_1, \ldots, A_n: B_n)$ where $B_i=A_i$. In this action, it is common knowledge that \emph{all agents in each group $A_i$ simultaneously share all their knowledge with all other agents in the same group $A_i$}. A very special case is $n=1$, which represents the \emph{resolution} action $!(A)$: it is common knowledge that all the agents in $A$ share their information with each other.\footnote{The special case $\mathbf{!({\cal A})}$ was introduced in \cite{Baltag2010}, and axiomatized in different contexts in \cite{Boddy2014,Goldbach2015,BBS2016}. The general case $\mathbf{!(G)}$ was studied (as ``\emph{group resolution}'') in \cite{AgotnesWang2017}.}

\smallskip

\par\noindent\textbf{Public Hacking}  $!H_{a:b}$ (as in the \emph{WikiLeaks} case). It is common knowledge that \textit{agent $a$ `hacks' agent $b$'s database, using her knowledge of $b$'s password} (represented by some variable $v_b$), \textit{and makes all $b$'s data public}. Everybody gets to see all $b$'s data (including the value of his password $v_b$), but only $a$ and $b$ knew the password beforehand. The model is the same as for public sharing $!b$, except for the precondition, which is $\it{pre}_{!}:=K_a v_b\wedge K_b v_b$: this event happens only if $a$ knew $b$'s password $v_b$ before the event.

\smallskip

\par\noindent\textbf{Semi-public Hacking} This time, it is common knowledge that \textit{agent $a$ hacks agent $b$'s database 
(using her prior knowledge of his password $v_b$) and can read all his data; but the others cannot read these data}. The event model is similar to the one for semi-public sharing $!(a:b)$ (and in particular, the access map is the same), except that the precondition is $K_a v_b\wedge K_b v_b$ (as in the case of public hacking $!H_{a:b}$).

\smallskip

\par\noindent\textbf{Semi-public Change of Password} $!(v_a:= F(v_a, v'_a))$. It is common knowledge that \emph{$a$ changes her password $v_a$ to a 
value $F(v_a, v'_a)$, that is a function of her current password $v_a$ and of another (secret) local variable $v'_a$}. The encryption function $F$ is common knowledge, but the values of the former password $v_a$ and of the other secret number $v'_a$ are known (hopefully) only by $a$. The event model is similar for $!\varphi$, except that the precondition is $\it{pre}_{!}= K_a v_a\wedge K_a v'_a$, and the postcondition for $v_a$ is $\underline{!}(v_a)= F(v_a, v'_a)$.

\smallskip

Using larger event models, we can represent various forms of \emph{private and semi-private data-exchanges}.

\smallskip

\par\noindent\textbf{Secret Hacking}  Agent $a$ may be hacking
$b$'s database iff she got hold of $b$'s \emph{password $v_b$}. Only the hacker ($a$) knows whether or not she succeeded to get $v_b$. We represent this event $SH_{a:b}= (\bE, !)$ using a model $\bE= (\{!, \tau\}, \sim, \underline{\bullet}(\bullet))$ with two events $!$ (for \emph{successful hacking}) and $\tau$ (\emph{unsuccessful hacking}). The preconditions are $\it{pre}_{!}:= K_a v_b \wedge K_b v_b$ and $\it{pre}_{\tau}:= \neg K_a v_b \wedge K_b v_b$. The access maps are $\underline{!}(a):=\{a,b\}$, $\underline{\tau}(a)=\{a\}$, and $\underline{!}(c)=\underline{\tau}(c)=\{c\}$ for all $c\not=a$; all postconditions are given by identity; $a$'s accessibility is the identity relation, and the others' relations are the universal relation.

\smallskip

\par\noindent\textbf{Conditional Change of Password} $!(K_a |v_a|_b^{true} \leq N /v_a:= F(v_a, v'_a))$. It is common knowledge that \emph{$a$ changes her password $v_a$} (as in the semi-public change) \emph{only iff she knows that $b$ has succeeded to narrow down her possible passwords to $N$ possibilities}. This is an event $(\bE, !)$, whose event model $\bE= (\{!, ?\}, \sim, \underline{\bullet}(\bullet))$ has two event names $!$ (for \emph{changing $a$'s password}) and $?$ (for \emph{no change}). The preconditions are $\it{pre}_{!}:= K_a |v_a|_b^{true} \leq N$ and $\it{pre}_{?}:= \neg K_a |v_a|_b^{true} \leq N$. The postconditions for $v_a$ are 
$\underline{!} (v_a):=  F(v_a, v'_a)$ and $\underline{?} (v_a):= v_a$, while all other postconditions and access map are given by identity: nobody gains access to others' databases (since the hacking is prevented by password-change). Finally, agent $a$'s accessibility is the identity relation, while all others' accessibility is the universal relation.

\smallskip

\par\noindent\textbf{Secret Detection of Hacking} We can modify the secret hacking example $SH_{a:b}$ to allow the possibility that $b$ might \emph{detect} $a$'s hacking attack (so that he might come to know that he is being hacked). Only $b$ knows whether he actually detects an attack. This event $(\bE, !)$ has a model $E=\{!,?,\tau\}$ with \emph{three} event names $!$ (detected hacking), $?$ (successful, undetected hacking) and $\tau$ (unsuccessful hacking). The access maps for $!$ and $\tau$ are as in the model for $SH_{a:b}$, while the access map for $?$ is the same as for $!$; and the same goes for $pre_!$, $pre_\tau$ and $pre_?$. Besides loops for all agents, we also have $!\sima ?$ (i.e. $a$ doesn't know whether his hacking is detected or not) and $? \simb \tau$ ($b$ can't distinguish between undetected hacking and unsuccessful hacking), while $\simc$ is the universal relation for all outsiders  $c\neq a,b$. 
\vspace{-4mm}

\section{Axiomatization and Decidability}\label{Axiomatization LDAE}
\vspace{-2mm}

We first look at the \emph{static fragment} $LDA$: its proof system $\mathbf{LDA}$ is in Table \ref{tb0}.

\begin{table}[h!]
\centering
\begin{tabular}{>{\raggedright\arraybackslash}p{4cm} >{\raggedright\arraybackslash}p{11cm}}
\hline
\textbf{(I)} & \textbf{Axioms and rules of classical propositional logic (CPL):}
\\
Modus Ponens rule & \& \,\, \, all instances of CPL axioms in the language of $LDAE$
\\
 \textbf{(II)} & \textbf{Axioms for equality and order}:
\end{tabular}
\\
\begin{tabular}{>{\raggedright\arraybackslash}p{4cm} >{\raggedright\arraybackslash}p{11cm}}
(Indiscernability) & $\ux=\uy \to (P\ux \uz \leftrightarrow P\uy\uz)$
\ \\
(Functionality) & $\ux=\uy \to F(\ux)=F(\uy)$
\ \\
(Definition by Cases) & $\varphi \, \to\, \left( x= (\varphi\!\!\to\!\! x|y)\right)$, \, \,\,\,\,\,\,\,\,\,\,\,\, 
$\neg\varphi \, \to \, \left(y= (\varphi\!\!\to\!\! x|y)\right)$
\ \\
(Transitivity) & $(x\leq y \wedge y\leq z)\, \to \, x\leq z$
\ \\
(Anti-symmetry) & $(x\leq y\wedge y\leq x) \, \to \, x=y$
\ \\
(Totality) & $x\leq y \vee y\leq x$
\ \\ 
(Top and Bottom) & $\bot \leq x\leq \top$
\ \\
(Non-trivial Constants) & $\top \neq \bot$
\end{tabular}
\\
\begin{tabular}{>{\raggedright\arraybackslash}p{4cm} >{\raggedright\arraybackslash}p{11cm}}
\textbf{(III)} & \textbf{Axioms and rules for (distributed) knowledge}:
\end{tabular}
\\
\begin{tabular}{>{\raggedright\arraybackslash}p{4cm} >{\raggedright\arraybackslash}p{11cm}}
 (Necessitation) &  From $\varphi$, infer $K_A \varphi$
 \ \\
(Distribution) & $K_A (\varphi\to \psi)\, \to \, (K_A \varphi\to K_A\psi)$
\ \\
(Veracity) & $K_A\varphi \, \to \, \varphi$
\ \\
(Local Knowledge) & $\varphi \, \to \, K_A\varphi$, \,\,\,\,\, \,\,\,\,\,\,\,\,\,\,\,\,\,\,\, for  $\mathcal{L}(\varphi)\subseteq A$
\end{tabular}
\\
\begin{tabular}{>{\raggedright\arraybackslash}p{4cm} >{\raggedright\arraybackslash}p{11cm}}
\textbf{(IV)} & \textbf{Axioms for (cutoff) minimum value}:
\end{tabular}
\\
\begin{tabular}{>{\raggedright\arraybackslash}p{4cm} >{\raggedright\arraybackslash}p{11cm}}
(Lower Bound) & $\varphi \, \to\, \mu_N x_A^\varphi \leq x$ 
\ \\
(Trivial Minimum) &
$K_A\neg\varphi \, \to \, \mu_N x_A^\varphi =\top
$
\ \\
(Undefined Minimum) & 
$|x|_A^\varphi>N \, \to \, \mu_N x_A^\varphi =\bot$
\\
(Reaching the Minimum) &
$\left(\varphi \wedge K_A^\varphi (x\!\in Y) \right)\, \to\,  \langle K_A^\varphi\rangle \left(\mu_N x_A^\varphi =x\right)$,
\,\,\,  for $|Y|\leq N$ s.t. $\mathcal{L}(Y)\subseteq A$
\ \\
\hline
\end{tabular}
\caption{The proof system  
$\mathbf{LDA}$ for the static fragment, using abbreviations 
$K_A^\theta\varphi$, $\langle K_A^\theta\rangle \varphi$, $x\in Y$, $|x|_A^\varphi> N$.
}\label{tb0}
\vspace{-2mm}
\end{table}

\begin{proposition}\label{theorems}
The following theorems are provable in the system $\mathbf{LDA}$:
\vspace{-1mm}\begin{enumerate}
\item (Propositional Introspection) \,\, \,\, $K_A\varphi \, \to \, K_A K_A\varphi$;\,\, \,\, $\neg K_A\varphi \, \to \, K_A\neg K_A\varphi$;
\item (Term Introspection) \,\, $K_A x$,\,\,\,\,\,\,\,\,\, \,\,\,\,\, for $\mathcal{L}(x)\subseteq A$;
\item (Group Monotonicity)  \,\, $K_A \, \varphi \to \, K_B \varphi$, \,\,\,\,\,\,\,\,\,\, \,\,\,\,\, for $A\subseteq B$.
\end{enumerate}
\end{proposition}

\begin{theorem}\label{CompLKV} (\emph{Soundness, Completeness and Decidability of $\mathbf{LDA}$})
The proof system $\mathbf{LDA}$
is complete for the static fragment $LDA$. Moreover, the static logic $LDA$ is decidable.
\vspace{-1mm}\end{theorem}
\begin{proof}
\vspace{-2mm}
\emph{Completeness} is shown in Section \ref{A}, using Prop \ref{theorems}. \emph{Soundness} is trivial, except for the 'Reaching the Minimum' axiom, whose soundness we sketch here. Suppose that $s\models \varphi \wedge K_A^\varphi (x\!\in Y)$ for some $Y\subseteq Var$ s.t. $|Y|\leq N$ and $\mathcal{L}(Y)\subseteq A$. To show that $s\models  \langle K_A^\varphi\rangle \left(\mu_N x_A^\varphi =x\right)$, we first prove the following:

\par\emph{Claim:} $Val(x_A^\varphi)_{s,\bM}\subseteq s(Y)_\bM$. 

To show this, let $d\in Val(x_A^\varphi)_{s,\bM}$, i.e. there is some $w\simA s$ s.t. $w\models\varphi$ and $w(x)=d$. Since $s\models K_A^\varphi (x\!\in Y)$, we infer that $w\models (x\!\in Y)$, hence $w(x)\in w(Y)$. But, given that $w\simA s$ and $\mathcal{L}(Y)\subseteq A$, we have that $w(Y)=s(Y)$ (by 
Proposition \ref{Preservation}). So we obtain $d=w(x)\in w(Y)=s(Y)$. Since this holds for any $d\in Val(x_A^\varphi)_{s,\bM}$, we conclude that $Val(x_A^\varphi)_{s,\bM}\subseteq s(Y)$, thus establishing our Claim.

Using the above Claim and the fact that $|Y|\leq N$, we have that $|Val(x_A^\varphi)_{s,\bM}|\leq N$. Since we also have that $s(x)\in Val(x_A^\varphi)_{s,\bM}\neq \emptyset$ (since $s\models\varphi$), we obtain that $s(\mu_N x_A^\varphi)= min_N Val(x_A^\varphi)_{s,\bM}= min\, Val(x_A^\varphi)_{s,\bM}\in Val(x_A^\varphi)_{s,\bM}$ (by the semantic clause for $\mu_N$ and the definition of $min_N$), and hence 
$s(\mu_N x_A^\varphi)= w(x)$ for \emph{some} $w\simA s$ with $w\models \varphi$. Applying again Proposition \ref{Preservation}, we have 
$w(\mu_N x_A^\varphi)= s(\mu_N x_A^\varphi)=w(x)$ (since $\mathcal{L}(\mu_N x_A^\varphi)=A$). This, together with  $w\simA s$ and $w\models \varphi$, yields $s\models  \langle K_A^\varphi\rangle \left(\mu_N x_A^\varphi =x\right)$, as desired.
\end{proof} 

\smallskip

As for the full \emph{dynamic logic} $LDAE$, we need a few more notations and results to state our axioms.

\smallskip
 
\par\noindent\textbf{Possible Values after an Event} Recall that $Val(x_A^\varphi)_{s, \bM} =\{w(x)_\bM: w\simA s, w\models_\bM \varphi\}$ is the set of possible values at state $s$ according to $A$ given $\varphi$. The following (easily checked) result gives us a characterization of the corresponding set of possible values of $x$ (according to $A$ given $\varphi$) \emph{after an event $\be$}:
\begin{lemma}\label{Post-Values}
$Val(x_A^\varphi)_{(s,\be), \bM\otimes \bE}= \bigcup_{f\simA e} Val \left(\mathbf{f}(x)_{\underline{f}(x)}^{\langle \mathbf{f} \rangle \varphi}\right)_{s,\bM}$
\end{lemma}

\par\noindent\textbf{Counting the Possible Values after an Event} We want a formula expressing the fact that \emph{there will be at most $N$ possible values of $x$} (according to $A$ given $\varphi$) \emph{after the event $\be$}. Moreover, we want to express this \emph{at the current state} (before the event). Put now $Var_e^N (x_A^\varphi) :=
\{n_N \mathbf{f}(x)_{\underline{f}(A)}^{\langle \mathbf{f}\rangle \varphi} : f\simA e, n\leq N\}$.
Semantically, given Lemma \ref{Post-Values}, it should be clear that, \textit{if $|Val(\mathbf{f}(x)_{\underline{f}(A)}^{\langle f\rangle \varphi})_{s,\bM}|\leq N$ holds for all events $f\simA e$, then $Val(x_A^\varphi)_{(s,\be), \bM\otimes \bE}\subseteq s(Var_e^N (x_A^\varphi))_\bM$}. 
However, \emph{this inclusion might be strict}: it can happen that \emph{not all of the values in  $s(Var_e^N (x_A^\varphi))_\bM$ are possible values} of $x$ according to $A$ (given $\varphi$, after the event $e$). The problem is that whenever $|Val(\mathbf{f}(x)_{\underline{f}(A)}^{\langle f\rangle \varphi})_{s,\bM}| < N$ 
for \emph{some }$f\simA e$, we get a possibly 'fake' $x$-value $N_N. \mathbf{f}(x)_{\underline{f}(A)}^{\langle \mathbf{f}\rangle \varphi}=\top_\bD$. So, for any $z\in Var_e^N (x_A^\varphi)$,
its condition of possibility is given by the formula:
\vspace{-2mm}
$$\Diamond z \,\, :=\,\, \left(z=\top \, \to\, \bigvee_{f\simA e} \langle K_{f(A)}^{\langle f\rangle\varphi} \rangle \mathbf{f}(x)=\top \right).
\vspace{-4mm}$$  
\vspace{-2mm}
Thus, for any \emph{subset $Z\subseteq Var_e^N (x_A^\varphi)$}, the formula
\vspace{-3mm}
$$|Z|^\Diamond \leq N \,\, :=\,\, \bigvee_{Y\subseteq Z, |Y|\leq N} \, \bigwedge_{z\in Z} \left( \Diamond z \to \bigvee_{y\in Y} z=y\right)
\vspace{-2mm}$$
says that \emph{the number of possible values of $x$ in $Z$ (according to $A$ given $\varphi$) is at most $N$}. Finally, by applying this to the whole set 
$Z= Var_e^N (x_A^\varphi)$ above, we obtain the desired formula:

\begin{lemma}\label{Counting}
Let $s$ be a state in a state model $\bM$, let $\be$ be an event in an event model $\bE$ with $s\models_\bM pre_\be$, and let $x$,$A$, $\varphi$ be s.t. $|Val(f(x)_{\underline{f}(A)}^{\langle \mathbf{f}\rangle \varphi})_{s,\bM}|\leq N$ holds for \emph{every} $f\simA e$. Then we have the equivalences:
\vspace{-2mm}
$$s\models_\bM |Var_e^N (x_A^\varphi)|^\Diamond\leq N \,\, \,\,\mbox{ iff }\,\, \,\, |Val(x_A^\varphi)_{(s,\be), \bM\otimes \bE}|\leq N \,\, \,\,\mbox{ iff } \,\, \,\,
(s,e)\models_{\bM\otimes\bE} |x_A^\varphi|\leq N.
$$
\end{lemma}

Using these notations, the proof system $\mathbf{LDAE}$ for our full dynamic logic is given in Table \ref{tb1}.
\begin{table}[h!]
\centering
\begin{tabular}{>{\raggedright\arraybackslash}p{3.5cm} >{\raggedright\arraybackslash}p{11.6cm}}
\hline
\textbf{(I)} & \textbf{Static Axioms and rules of $\mathbf{LDA}$}
\\
All $\mathbf{LDA}$ rules \,\, \& \, all & $\mathbf{LDA}$ axiom schemas extended to formulas of the full language $LDAE$
\ \\
\end{tabular}
\\
\begin{tabular}{>{\raggedright\arraybackslash}p{3.5cm} >{\raggedright\arraybackslash}p{11.5cm}}
 \textbf{(II)} & \textbf{Reduction axioms and rules for prop.} {\bf formulas}:
\end{tabular}
\\
\begin{tabular}{>{\raggedright\arraybackslash}p{3.5cm} >{\raggedright\arraybackslash}p{11.5cm}}
($[e]$-Necessitation) &  From $\varphi$, infer $[\be]\varphi$
 \ \\
 (Change of Facts) 
  & $[\be] p \, \, \leftrightarrow \,\, \left( pre_\be \to 
  post_\be (p)\right)$
 \ \\
 (Change of Properties) 
 & $[\be] P\ux \,\, \leftrightarrow \,\, \left( pre_\be \to P\be(\ux) \right)$
\ \\
(Distributivity) & $[\be] (\varphi\to \psi) \, \leftrightarrow \, \left([\be]\varphi\to [\be]\psi\right)$
\ \\
(Knowledge Update) & $[\be] K_A\varphi \, \, \leftrightarrow \,\, \left( pre_\be \to \bigwedge \{ K_{\underline{e}(A)} [{\mathbf{f}}] \varphi : f\!\!\simA\! e\}\right)$
\ \\
\end{tabular}
\\
\begin{tabular}{>{\raggedright\arraybackslash}p{3.5cm} >{\raggedright\arraybackslash}p{11.5cm}}
\textbf{(III)} & \textbf{Reduction axioms for data terms}:
\end{tabular}
\\
\begin{tabular}{>{\raggedright\arraybackslash}p{3.5cm} >{\raggedright\arraybackslash}p{11.5cm}}
(Basic Value Change) & $\be(v)\, \, =\, \, \left(pre_\be\!\!\to\!\!
post_\be(v)\right)$ \ \\
(Functional Change) & $\be(F(\ux))  \, \, = \, \, \left(pre_\be\!\!\to\!\! F(\be(\ux))  \right)$ \ \\
(Change of Cases) & $\be(\varphi\!\!\to\!\!x|y) \,\,  =\, \,  \left( [\be]\varphi\!\!\to\!\! \be(x)|\be(y)   \right)$
\ \\
(Cutoff-Min. Change) 
 & $\be(\mu_N x_A^\varphi) \, \, =\, \,
\left(pre_\be\!\!\to\!\! \left(|Var_e^N (x_A^\varphi)|^\Diamond \leq N \!\!\to\!\! 
min \{\mu_N {\mathbf{f}}(x)_{\underline{e}(A)}^{\langle {\mathbf{f}}\rangle \varphi}:f\!\!\simA\! e\}
|\bot \right) \right)$
\ \\ 
\hline
\end{tabular}
\caption{The proof system  
$\mathbf{LDAE}$, 
where $\be=(\bE,e), {\mathbf{f}}=(\bE,f)\in Events$.
}\label{tb1}
\vspace{-2mm}
\end{table}

\begin{proposition}\label{theorems2}
The following derived reduction laws are provable in the system $\mathbf{LDAE}$:
\begin{itemize}
\item (Impossible Change)\, $[\be] false \, \leftrightarrow \, \neg pre_e$; and $[\be] true \,\leftrightarrow \, true$;
\item (Negation \& Conjunction Reduction)\, $[\be] \neg\varphi \, \leftrightarrow \,  \left(pre_e\to \neg [\be]\varphi \right)$; and
$[\be](\varphi\wedge \psi) \, \leftrightarrow \,  \left([\be]\varphi\wedge [\be]\psi\right)$;
\item\, (Preservation of Constants) $\be(\top) = \top$; and $\be(\bot) \, =\, (pre_e\!\! \to\!\! \bot)$;
\item\,  (Minimum Reduction) $\be(min\, X) \,=\, min\, \{\be(x): x\in X\}$.
\end{itemize}
\end{proposition}

\begin{theorem}\label{CompDLKV} (Soundness, Completeness, Expressivity and Decidability of $\mathbf{LDAE}$)
The proof system $\mathbf{LDAE}$ in Table \ref{tb1} is sound and complete for the dynamic logic $LDAE$. Moreover, $LDAE$ is provably co-expressive with its static fragment $LDA$, and thus it is decidable. 
\vspace{-1mm}\end{theorem}
\begin{proof}
\vspace{-2mm}
\emph{Completeness} is shown in Section  \ref{B}. \emph{Soundness} is trivial, except for the `Cutoff-Min' reduction axiom, which we sketch here. Let $\bM$ be any state model, and $s$ be any state. We consider \emph{two cases}:
\medskip

\noindent\emph{Case 1:} $s\not\models pre_\be$. In this case, both terms of the equality claimed in the axiom take value $\top_\bS$ at $s$ in $\bM$.
\smallskip

\noindent\emph{Case 2:}  $s\models pre_\be$. In this case, $s(\be(\mu_N x_A^\varphi))_\bM= (s,e)(\mu_N x_A^\varphi))_{\bM\otimes \bE}= min_N \left(Val(x_A^\varphi)_{(s,e),\bM\otimes \bE} \right)$, and there are \emph{two subcases} to consider:
\medskip

\noindent\emph{Subcase (2A)}: there exists some $f^0\in E$ s.t. $f^0\simA e$ and $|Val(\mathbf{f^0}(x)_{\underline{f^0}(A)}^{\langle \mathbf{f^0}\rangle \varphi})_{s,\bM}|> N$. Then we have by definition that $s(\mu_N {\mathbf{f^0}}(x)_{\underline{e}(A)}^{\langle {\mathbf{f^0}}\rangle \varphi})=\bot_\bD$ (given that 
$f_0\simA e$ implies $\underline{f^0}(A)= \underline{e}(A)$), and thus $min \{\mu_N {\mathbf{f}}(x)_{\underline{e}(A)}^{\langle {\mathbf{f}}\rangle \varphi}:f\!\!\simA\! e\}$ takes value $\bot_\bD$ at $s$. So the right-hand side term of the equality in the 'Cutoff-Min' reduction axiom takes value $\bot_\bD$ at $s$ (regardless of whether we have $s\models_\bM |Var_e^N (x_A^\varphi)|^\Diamond\leq N$ or not). On the other hand, the left-hand side also evaluates to $\bot_\bD$ at $s$ (since by Lemma \ref{Post-Values}, $|Val(\mathbf{f^0}(x)_{\underline{f^0}(A)}^{\langle \mathbf{f^0}\rangle \varphi})_{s,\bM}|> N$ implies that $|Val(x_A^\varphi)_{(s,\be), \bM\otimes \bE}|> N$, hence $ (s,e)(\mu_N x_A^\varphi))_{\bM\otimes \bE}= min_N \left(Val(x_A^\varphi)_{(s,e),\bM\otimes \bE} \right)=\bot_\bD$), as desired.
\smallskip

\noindent\emph{Subcase (2B)}: we have $|Val(\mathbf{f}(x)_{\underline{f}(A)}^{\langle \mathbf{f}\rangle \varphi})_{s,\bM}|\leq N$ for all $f\in E$ s.t. $f\simA e$. We are in the conditions of Lemma \ref{Counting}, and there are again \emph{two subcases} to consider:
\smallskip

\noindent\emph{Subcase (2B1)}: $s\not\models_\bM |Var_e^N (x_A^\varphi)|^\Diamond\leq N$. In this case, we can use Lemma \ref{Counting} to check that both sides of the equality in the 'Cutoff-Min' reduction axiom evaluate to $\bot_\bD$ at $s$.
\smallskip

\noindent\emph{Subcase (2B2)}: $s\models_\bM |Var_e^N (x_A^\varphi)|^\Diamond\leq N$. In this case, we can use Lemma \ref{Post-Values} to show that  $s(\be(\mu_N x_A^\varphi))_\bM= (s,e)(\mu_N x_A^\varphi))_{\bM\otimes \bE}= min_N \left(Val(x_A^\varphi)_{(s,e),\bM\otimes \bE} \right)=
 min_N \left(\bigcup_{f\simA e} Val \left(\mathbf{f}(x)_{\underline{f}(x)}^{\langle \mathbf{f} \rangle \varphi}\right)_{s,\bM} 
\right)$. So the left-hand side of the equality in the 'Cutoff-Min' axiom evaluates at state $s$ (in $\bM$) to $min_{f\simA e}\, min_N  Val \left(\mathbf{f}(x)_{\underline{f}(x)}^{\langle \mathbf{f} \rangle \varphi}\right)_{s,\bM}$. On the other hand, we can use Lemma \ref{Counting} to check that the right-hand side of the equality in the 'Cutoff-Min' reduction axiom evaluates to the same expression at $s$.
\vspace{-2mm}\end{proof}
\vspace{-4mm}
\section{Completeness and Decidability Proofs for $LDA$}\label{A}
\vspace{-2mm}
Throughout this section, \emph{we fix a formula $\varphi_0\in Fml$}. We prove Theorem \ref{CompLKV} by
the method of \emph{quasi-models}. But, to obtain an appropriate analogue of Fischer-Ladner closure, we need two auxiliary notions:

\smallskip
\par\noindent\textbf{Restricted Vocabulary} For any finite set $\Sigma\subseteq Fml$, the \emph{$\Sigma$-restricted vocabulary} 
${\mathcal{V}_\Sigma := (\Agents_\Sigma, V_\Sigma, Prop_\Sigma,}$ ${Pred_\Sigma, Funct_\Sigma, ar_\Sigma, \epsilon)}$ is formed as follows: 
$\Agents_\Sigma$ is the set of agents occurring (inside terms or modalities in formulas) in $\Sigma$;
$V_\Sigma$ is \emph{the set of basic variables occurring (as subterms of any term) in formulas of $\Sigma$}; 
$Prop_\Sigma: = Prop\cap \Sigma$ is the set of atomic propositions in $\Sigma$;   
$Pred_\Sigma$ is the set of predicates occurring in (formulas of) $\Sigma$; 
$Funct_\Sigma$ is the set of function symbols occurring in (formulas of) $\Sigma$; 
$ar_\Sigma$ is the restriction of $ar$ to $Pred_\Sigma\cup Funct_\Sigma$.
Clearly, if $\Sigma$ is finite, then all the sets in $\mathcal{V}_\Sigma$ are finite.

\smallskip
\par\noindent\textbf{The $\Sigma$-Restricted Set of Terms} For any finite set of formulas $\Sigma\subseteq Fml$, the \emph{$\Sigma$-restricted set of terms}
$Var_\Sigma$ is the smallest set of terms 
satisfying the following closure conditions: $Var_\Sigma$ contains $\bot$ and $\top$, as well as all terms occurring in any formula of $\Sigma$ (hence $V_\Sigma\subseteq Var_\Sigma$); 
$Var_\Sigma$ is closed under subterms;
if $i_N.x_A^\varphi\in Var_\Sigma$ for some $i\leq N$, then  $j_N.x_A^\varphi\in Var_\Sigma$ for all $j\leq N$.
Note that $Var_\Sigma$ is only \emph{a finite subset} of the (typically infinite) set $Var(\mathcal{V}_\Sigma)$ of terms of the language $LDA(\mathcal{V}_\Sigma)$.

\medskip

We now proceed to introduce the appropriate notions of Fisher-Ladner Closure, syntactic types, and special sets of types called quasi-models.

\smallskip

\par\noindent\textbf{Fisher-Ladner Closure} Given now our fixed formula $\varphi_0$, the \emph{closure of $\varphi_0$} is the smallest set of formulas $\Sigma=\Sigma(\varphi_0)$ satisfying the following closure conditions:
$\varphi_0\in \Sigma$; 
$true, false\in \Sigma$;
$\Sigma$ is closed under subformulas and single negations $\sim\varphi$; 
if $P\in Pred_\Sigma$ has arity $ar(P)=n$ and $\ux=(x_1, \ldots, x_n)$ is an $n$-tuple with all $x_1, \ldots, x_n\in Var_\Sigma$, then $P\ux\in \Sigma$;
if $(K_A\varphi)\in \Sigma$, $\theta$ is a subformula of $\sim\varphi$ and $B\subseteq\Agents_\Sigma$, then $(K_B\theta)\in \Sigma$ (hence also 
$\langle K_A \rangle \theta\in \Sigma$);
if $(\varphi\!\!\to\!\!x|y)\in Var_\Sigma$, then $\varphi\in \Sigma$;
if $x,y\in Var_\Sigma$, then $(x=y), (x\leq y)\in \Sigma$;
if $(\mu_N x_A^\varphi), z\in Var_\Sigma$ and $Y\subseteq Var_\Sigma$, 
then $K_A^\varphi(z\in Y), \langle K_A^\varphi\rangle (z\in Y) \in\Sigma$.

\medskip

\par\noindent\textbf{Types} Let $\Sigma$ be the closure of $\varphi_0$. A \emph{$\Sigma$-type} is a subset of $\Sigma$ with the following properties:
\begin{enumerate}
\item\label{neg} for every $\varphi\in\Sigma$: $(\sim\varphi)\in \Delta$ iff $\varphi\not\in\Delta$;
\item\label{conj} for every $(\varphi\wedge \psi)\in \Sigma$: $(\varphi\wedge \psi)\in \Delta$ iff $\varphi\in\Delta$ and $\psi\in \Delta$;
\item\label{eq4} for every  $P\ux\uz\in \Sigma$: if $(\ux=\uy), P\ux\uz\in \Delta$ then $P\uy\uz\in \Delta$;
\item\label{eq5} for $F(\ux), F(\uy)\in Var_\Sigma$: if  $(\ux=\uy)\in \Delta$, then $(F(\ux)=F(\uy))\in \Delta$;
\item\label{trans} if $(x\leq y), (y\leq z)\in \Delta$, then $(x\leq z)\in \Delta$;
\item\label{antisymm} if $(x\leq y), (y\leq x)\in \Delta$, then $(x=y)\in \Delta$;
\item\label{total} either for all $x, y\in Var_\Sigma$, we have either $(x\leq y)\in \Delta$ or $(y\leq x)\in \Delta$;
\item\label{bounds} for all $x\in Var_\Sigma$, we have $(\bot\leq x), (x\leq \top)\in \Delta$;
\item\label{non-trivial} $(\top\neq\bot)\in \Delta$;
\item\label{cases} for every $(\varphi\!\!\to\!\! x|y)\in\Sigma$: $\varphi\in \Delta$ implies $(x=(\varphi\!\!\to\!\! x|y))\in\Delta$; and $\varphi\not\in\Delta$ implies $(y=(\varphi\!\!\to\!\! x|y))\in \Delta$;
\item\label{prop-veracity} if $(K_A\varphi)\in\Delta$ then $\varphi\in\Delta$;
\item\label{Lower Bound} for all $\mu_N x_A^\varphi\in Var_\Sigma$: $(\mu_N x_A^\varphi\leq x)\in \Delta$;
\item\label{Reaching the Min} if $\varphi, K_A^\varphi (\mu_N x_A^\varphi \neq x)\in \Delta$ and $Y\subseteq Var_\Sigma$ is s.t. $\mathcal{L}(Y)\subseteq A$ and $|Y|\leq N$, then $\langle K_A^\varphi\rangle (x\not\in Y)\in \Delta$;
\item\label{Trivial Min} for all $\mu_N x_A^\varphi\in Var_\Sigma$: if $(K_A\neg\varphi)\in \Delta$, then
$(\mu_N x_A^\varphi=\top)\in \Delta$;
\item\label{Undefined Min}  for all $\mu_N x_A^\varphi\in Var_\Sigma$:  
if $(|x|_A^\varphi>N)\in \Delta$, then $(\mu_N x_A^\varphi=\bot)\in \Delta$.   
        \end{enumerate}

        \par\noindent\textbf{Observation} \emph{Types are closed under modus ponens}: if $\Delta$ is a type and $(\varphi\to\psi), \varphi\in \Delta$, then $\psi\in \Delta$.
        
        \medskip

\par\noindent\textbf{Accessibility relations on types}
For types $\Delta, \Delta'$ and group $A\subseteq \Agents_\Sigma$, we put:
\vspace{-2mm}
$$\Delta \simA \Delta' \,\, \mbox{ iff } \,\, \varphi\in \Delta \Leftrightarrow \varphi\in\Delta' \, \mbox{ for all } \varphi\in\Sigma \mbox{ s.t. } \mathcal{L}(\varphi)\subseteq A.
\vspace{-3mm}$$
\vspace{-2mm}
\begin{proposition}\label{Access}
\vspace{-2mm}
The relations $\simA$ are equivalence relations on types, satisfying Group Monotonicity:
$\Delta\simA \Delta' \mbox{ and } B\subseteq A \mbox{ imply } \Delta\simB \Delta'.$
\vspace{-2mm}
\end{proposition} \vspace{-1mm}
\begin{proof} \vspace{-1mm}
This follows directly from the definition of relations $\simA$ on types.
\end{proof}
\vspace{-1mm}
\begin{proposition}\label{Access2}
If $\Delta\simA \Delta'$ and $(K_A\varphi)\in \Delta$, then $\varphi\in \Delta'$.
\vspace{-1mm}
\end{proposition}
\begin{proof}
\vspace{-1mm}
Since $\mathcal{L}(K_A\varphi)=A$ and $\Delta\simA \Delta'$, we use the definition of $\simA$ on types and the fact that
$(K_A\varphi)\in \Delta$ to infer that $(K_A\varphi)\in \Delta'$. This together with condition \ref{prop-veracity} on types, gives us the desired conclusion.
\end{proof}

\smallskip

\par\noindent\textbf{Hat notation}. For every type $\Delta$, we put
$\widehat{\Delta} \, \,\,\, :=\,\, \, \, \bigwedge\Delta$ for the \emph{conjunction of all formulas in $\Delta$}.

\begin{proposition}\label{Access3}
Given types $\Delta$ and $\Lambda$, if $\widehat{\Delta}\wedge \langle K_A \rangle \widehat{\Lambda}$ is consistent, then $\Delta \simA \Lambda$.
\vspace{-2mm}
\end{proposition}
\vspace{-1mm}
\begin{proof} 
\vspace{-1mm} 
Let $\Delta$ and $\Lambda$ be types as above, and suppose towards a contradiction that we have $\Delta \not\simA \Lambda$. Then there must exist $\varphi\in \Sigma$ s.t. $\mathcal{L}(\varphi)\subseteq A$ and $\varphi\in\Delta$, but $(\sim\!\varphi)\in\Lambda$. From this together with the assumption that $\widehat{\Delta}\wedge \langle K_A \rangle \widehat{\Lambda}$ is consistent, we infer that $\varphi\wedge \langle K_A \rangle \neg\varphi$ is consistent, and thus $\varphi\wedge \neg K_A\varphi$ is consistent. But this contradicts the fact that $\vdash\, \varphi\to K_A\varphi$ is an $\mathbf{LDA}$-theorem for formulas $\varphi$ with $\mathcal{L}(\varphi)\subseteq A$.
\end{proof}

\medskip 

\par\noindent\textbf{Quasi-Models}
A \emph{quasi-model for $\varphi_0$} is a set $S$ of types over $\Sigma$, with the following two properties:
(*) $\varphi_0\in \Delta_0$ for some type $\Delta_0\in S$; (**) if $\langle K_A \rangle \varphi\in \Delta\in S$, then there is some $\Delta'\in S$ with $\Delta\simA \Delta'$ and $\varphi\in \Delta'$.

\begin{proposition}\label{Quasi-model properties}
If $\Delta\in S$ is a type in a quasi-model $S$, then we have the following:
\begin{description}
\item[(1)] if $K_A(\varphi\to\psi)\in \Delta$ and $(K_A\varphi)\in \Delta$, then $K_A\psi\in \Delta$;
\item[(2)] for every $(K_A\varphi)\in \Sigma$ s.t. $\mathcal{L}(\varphi)\in A$, if $\varphi\in \Delta$ then $(K_A\varphi)\in \Delta$;
\item[(3)] if $(K_A\varphi)\in \Delta$ and $B\subseteq A$, then $(K_B \varphi)\in\Delta$.
\end{description}
\end{proposition}

\smallskip 

\par\noindent\textbf{The $\Sigma$-canonical quasi-model} A \emph{$\Sigma$-theory} is any maximally consistent subset of $\Sigma$.  We denote by
$S_\Sigma$ the \emph{set of all $\Sigma$-theories}. We will \emph{show that $S_\Sigma$ is a (finite) ``canonical'' quasi-model} for $\Sigma$.

\begin{lemma}\label{Types}
Every $\Sigma$-theory is a $\Sigma$-type.
\end{lemma}
\begin{proof}
\vspace{-1mm}
This is easy to check, using the axioms of $\mathbf{LDA}$ and Proposition \ref{theorems}.
\end{proof}

\vspace{-2mm}
\begin{lemma}\label{K-ExistenceLemma}
For every $\Delta\in S_\Sigma$, if $\langle K_A \rangle \varphi\in \Delta$ then there is some $\Delta'\in S_\Sigma$ with $\Delta\simA \Delta'$ and $\varphi\in \Delta'$.
\vspace{-1mm}
\end{lemma}
\vspace{-1mm}
\begin{proof}
\vspace{-1mm}
Put
$\Delta_A\, :=\, \{\theta: \theta\in\Delta \mbox{ s.t. } \mathcal{L}(\theta)\subseteq A\}
\cup \{\sim \theta: \theta\in (\Sigma -\Delta) \mbox{ s.t. } \mathcal{L}(\theta)\subseteq A\}$.

\emph{Claim}: $\Delta_A \cup \{\varphi\}$ is consistent wrt the system $\mathbf{LDA}$.

\emph{Proof of Claim}: Suppose not. Then we have $\vdash \widehat{\Delta_A}\to \sim \varphi$. Applying Necessitation and Distribution, we obtain $\vdash K_A\widehat{D_A} \to K_A\sim \varphi$. On the other hand, by inspecting the structure of $\Delta_A$, it is easy to see that $\mathcal{L}(\Delta_A)\subseteq A$, so by Strong Introspection (Proposition \ref{theorems}) we have $\vdash \widehat{\Delta_A} \to K_A \widehat{\Delta_A}$. Putting these together, we get $\vdash \widehat{\Delta_A}\to K_A\sim \varphi$. Since $\Delta_A\subseteq \Delta$ and $\Delta$ is closed under $\Sigma$-consequences, we have $(K_A\sim\varphi)\in \Delta$. But this contradicts the assumption that $\langle K_A \rangle \varphi\in \Delta$ (given the consistency of $\Delta$).
\\
Using our Claim and the standard Lindenbaum Lemma, we get that $\Delta_A \cup \{\varphi\}$ has a $\Sigma$-maximally consistent extension $\Delta'\in S_\Sigma'$. So we have $\varphi\in \Delta'$ and $\Delta_A\subseteq \Delta'$, which implies that $\Delta\sim_A \Delta'$.
\end{proof}

\vspace{-2mm}
\begin{proposition}\label{QuasiModel}
If $\varphi_0\in \Sigma$ is consistent, then there exists a quasi-model for $\varphi_0$.
\vspace{-1mm}
\end{proposition}
\vspace{-1mm}
\begin{proof}
\vspace{-1mm}
Take the set $S_\Sigma$ of all $\Sigma$-theories. By the Lindenbaum Lemma, there exists a maximally consistent subset $\Delta_0\in S_\Sigma$, such that $\varphi_0\in \Delta_0$. By Lemmas \ref{Types} and \ref{K-ExistenceLemma},
$S_\Sigma$ is a quasi-model for $\varphi_0$.
\end{proof}

\par\noindent
\begin{corollary} \emph{If $\varphi_0$ is satisfiable then there exists a quasi-model for $\varphi_0$.}
\end{corollary}

\medskip

The hard part is to prove the \emph{converse} of this:

\begin{proposition}\label{Satisfiability}
If $S$ is a quasi-model for $\varphi_0$, then $\varphi_0$ is satisfiable.\footnote{Note that \emph{this would immediately give us the completeness and decidability of $\mathbf{LDA}$} (by Proposition \ref{Satisfiability}, Proposition \ref{QuasiModel} and the soundness of the system $\mathbf{LDA}$, as well as the finiteness of the closure $\Sigma=\Sigma(\varphi_0)$, and the decidability of checking that a subset of $\Sigma$ is a quasi-model).}
\end{proposition}

\emph{The rest of this section is dedicated to the proof of Proposition \ref{Satisfiability}}.

\medskip

\par\noindent\textbf{Unravelling: the tree of histories} Let us fix a quasi-model $S$, a formula $\varphi_0$ and a type $\Delta_0\in S$ with $\varphi_0\in \Delta_0$. We will construct a model for $\varphi_0$, based on an unravelling of $S$ around $\Delta_0$.
A \emph{history} is a finite sequence $h=(\Delta_0, A^1, \Delta_1, \ldots, A^n, \Delta_n)$ of any length $n\geq 0$, where $\Delta_1, \ldots,\Delta_n\in S$ are types and
$A^1, \ldots, A^n\subseteq\Agents_\Sigma$ are groups, such that we have
$\Delta_{i-1}\sim_{A^i} \Delta_i$ for all $i=1,n$. Let $H$ be the set of all histories.
We denote by $last(h):=\Delta_n$ the last state in history $h$, and by $\to_{A}$ the natural \emph{forward one-step relation} on histories in $H$, given by putting: $h\to_A h'$ iff $h'=(h, A, \Delta')$ (with $last(h) \simA \Delta'=last(h')$). We denote by $\ot_{A}$ the \emph{backward one-step relation}, defined as the converse of the forward relation: $h \ot_A h'$ iff $h'\to_A h$.
The one-step relations structure $H$ into a \textit{tree rooted at $\Delta_0$} (with the immediate successor relation given by $h\to h'$ iff $h\to_A h'$ for some group $A$). In particular, we have \textit{the tree property} : \emph{every two nodes $h, h'$ of the tree are connected by a unique non-redundant path} $h=h_0 \ot_{A^1} h_1 \ot_{A^2} \ldots \ot_{A^i} h_{i}\to_{A^{i+1}} \ldots \to_{A^n} h_n=h'$ (-in which neighboring nodes are immediate successors, in one order or another, and no nodes are repeated).

\medskip
\par\noindent\textbf{Epistemic relations on histories}
To make this tree into a model for our restricted vocabulary $\mathcal{V}_\Sigma$, we define our \emph{single-agent indistinguishability relations} $\sima\subseteq H\times H$ on histories, by putting
\vspace{-2mm}
$$\sima \,\, :=\,\, \left( \bigcup_{A\ni a} \to_A \cup \bigcup_{A\ni a} \ot_A \right)^*,
\vspace{-2mm}$$
where $\ot_A$ is the converse of $\to_A$, the unions range over groups $A\subseteq \Agents_\Sigma$ s.t. $a\in A$, and $R^*$ is the reflexive-transitive closure of $R$.
Since our goal is to build a standard model, the \emph{group indistinguishability relations} $\simA\subseteq H\times H$ are taken to be simply the \emph{intersections} $\simA \,\, :=\,\, \bigcap_{a\in A} \sima$ of all the individual relations.

It is useful to give more concrete characterizations of the relations $\simA$ (and $\sima$) on histories:

\begin{lemma}\label{path}
The following are \emph{equivalent}, for $A\subseteq \Agents_\Sigma$ and histories $h,h'\in H$:
\begin{enumerate}
\item $h\simA h'$;
\item $A\subseteq A^i$, for all groups $A^i$ that appear as transition labels on the non-redundant path from $h$ to $h'$.
\end{enumerate}
\vspace{-2mm}
\end{lemma}
\vspace{-1mm}
\begin{proof}
\vspace{-2mm}
Use the definitions of $\sima$ and $\simA=\bigcap_{a\in A}\sima$ on $H$, and the uniqueness of non-redundant path.
\end{proof}

\begin{lemma}\label{equiv}
If $h\simA h'$, then $last(h)\simA last(h')$.
\vspace{-1mm}
\end{lemma}
\vspace{-1mm}
\begin{proof}
\vspace{-1mm}
The proof is by \emph{induction on the length $N$ of the non-redundant path} from $h$ to $h'$.

For the \emph{base case} $h=h'$, the conclusion follows trivially (given that $\simA$ are equivalence relations).

\emph{Inductive case}: Suppose the non-redundant path from $h$ and $h'$ has length $N+1$, and let us look at the last transition on this path. Given Lemma \ref{path}, this transition can be either of the form
$h_N {\to}_{A^N} h_{N+1}=h'$, or of the form $h_N {\ot}_{A^N} h_{N+1}=h'$, with $A^N\supseteq A$. By definition of $\to_A$ on histories, we have either  $h'=(h_N, A^N, last(h'))$ or  $h_N=(h', A^N, last(h_N))$, with $last(h_N)\sim_{A^N} last (h')$ in both cases. By Monotonicity (Proposition \ref{Access}) and the fact that $A\subseteq A^N$, we obtain $last(h_N)\simA last(h')$. On the other hand, we also have $last(h)\simA last(h_N)$ (-since the non-redundant path from $h$ to $h_N$ has length $N$, so by the induction hypothesis the pair $(h,h_N)$ satisfies the conclusion of our Lemma, with $h'$ replaced by $h_N$). Putting these two together (and using the transitivity of $\simA$), we conclude that $last(h)\simA last(h')$, as desired.
\end{proof}

\begin{lemma}\label{K}
If $(K_A\varphi) \in last(h)$ and $h\simA h'$, then $\varphi\in last(h')$.
\vspace{-1mm}
\end{lemma}
\vspace{-1mm}
\begin{proof}
\vspace{-1mm}
By Lemma \ref{equiv}, $h\simA h'$ implies $last(h)\simA last(h')$. This, together with $K_A\varphi \in last(h)$, implies by Proposition \ref{Access2} that 
$\varphi\in last(h')$, as desired.
\end{proof}

\begin{lemma}\label{Diamond}(\emph{Diamond Lemma})
 If $\langle K_A \rangle \varphi\in last(h)$, then there exists some $h'\simA h$ s.t. $\varphi\in last(h')$.
 \vspace{-1mm}
 \end{lemma}
 \begin{proof}
\vspace{-1mm} By Lemma \ref{K-ExistenceLemma}, there exists some type $\Delta'\simA last(h)$ s.t. $\varphi\in \Delta'$. Take $h'=(h, A, \Delta')$. By Lemma \ref{path} we have $h'\simA h$, and obviously $\varphi\in \Delta'= last(h')$, as desired.\end{proof}

\begin{corollary}\label{KK}
If $(K_A\varphi)\in \Sigma$ and $h\in H$, then:  $(K_A\varphi) \in last(h)$ iff we have $\varphi\in last(h')$ for all $h'\simA h$.
\end{corollary} 

\smallskip

\par\noindent\textbf{The Value Domain: a Quotient Construction} The \emph{set of values} $D$ of our model will be a quotient of the Cartesian product $H\times Var_\Sigma$.
We \emph{define an equivalence relation}
$\approx$ on pairs $(history, variable)$ in $H\times Var_\Sigma$ (telling us \emph{when two such pairs represent the same value}), 
as well as \emph{a total preorder} $\lessapprox$ on these pairs in $H\times Var_\Sigma$ (telling us \emph{when the value of a pair is at most equal to another pair's value}).  
Then we take our canonical set of objects $D$ to be \emph{the quotient of $H\times Var_\Sigma$ with respect to $\approx$}, while the preorder $\lessapprox$
induces our desired \emph{total order} $\leq$ on the quotient $D$.

For this, we first introduce another \emph{equivalence relation} $\sim$ on $H\times Var_\Sigma$ (representing \emph{identity of objects at a given node}),
and another partial preorder $\lesssim$ on $\sim$ on $H\times Var_\Sigma$ (representing the \emph{order relation on values at a given node}).
This is given by putting:
\vspace{-2mm}
$$(h,x) \sim (h', x') \,\, \mbox{ iff } \,\, h=h' \mbox{ and } (x= x')\in last(h),
\vspace{-2mm}$$
$$(h,x) \lesssim (h', x') \,\, \mbox{ iff } \,\, h=h' \mbox{ and } (x\leq x')\in last(h).$$

\par\noindent Second, we define \emph{(forward and backward) one-step relations $\to_=$ and $\to_\leq$ on pairs} in $H\times Var_\Sigma$:
\vspace{-2mm}
$$(h,x) \to_= (h', x') \,\, \mbox{ iff } \,\, \exists y\in Var_\Sigma\, \exists A\supseteq \mathcal{L}(y)
\mbox{ s.t. } h\to_A h',
(x=y)\in last(h)  \mbox{ \& } (y=x')\in last(h');\vspace{-2mm}$$
$$(h,x) \ot_= (h', x') \,\, \mbox{ iff } \,\, \exists y\in Var_\Sigma\, \exists A\supseteq \mathcal{L}(y)
\mbox{ s.t. } h\ot_A h',
(x=y)\in last(h)  \mbox{ \& } (y=x')\in last(h');$$
$$(h,x) \to_\leq (h', x') \,\, \mbox{ iff } \,\, \exists y\in Var_\Sigma\, \exists A\supseteq \mathcal{L}(y)
\mbox{ s.t. } h\to_A h',
(x\leq y)\in last(h)  \mbox{ \& } (y\leq x')\in last(h');$$
$$(h,x) \ot_\leq (h', x') \,\, \mbox{ iff } \,\, \exists y\in Var_\Sigma\, \exists A\supseteq \mathcal{L}(y)
\mbox{ s.t. } h\ot_A h',
(x\leq y)\in last(h)  \mbox{ \& } (y\leq x')\in last(h').$$
Note that $\ot_=$ is just the converse of $\to_=$, but $\ot_\leq$ is \emph{not} the converse of $\to_\leq$.

\medskip

\par\noindent\textbf{Value Identity and Order}
Finally, \emph{\emph{we define our main equivalence relation $\approx$ and our total preorder $\lessapprox$ on pairs}} $(history, variable)$ in $H\times Var_\Sigma$, by putting: $\approx \, \,\, :=\, \,\, (\sim\cup\to_=\cup \ot_=)^*$; and $\lessapprox \, \,\, :=\, \,\, (\lesssim\cup \to_\leq \cup \ot_\leq)^*$, 
where $R^*$ is the reflexive-transitive closure of a relation $R$.
It is useful to have a more concrete characterization of $\approx$ and $\lessapprox$, in terms of the non-redundant path from $h$ to $h'$:

\begin{lemma} \label{PathLemma} (``Path Lemma'')
 Let $(h, x), (h', x')\in  H\times Var_\Sigma$, and let
$h=h_0 \ot h_1 \ot \ldots \ot h_i \to\ldots \to h_n=h'$
be the non-redundant path from $h$ to $h'$. Then the following are equivalent:
\begin{itemize}
\item $(h,x)\approx(h',x')$;
\item either $(h,x)\sim (h', x')$ (if $n=0$), or else there exist $x_0, x_1, \ldots, x_i, \ldots, x_n\in Var_\Sigma$ s.t.
we have:
$(h,x)= (h_0, x_0)\leftarrow_= (h_1, x_1) \leftarrow_= (h_2, x_2)\leftarrow_=\cdots \leftarrow_= 
(h_i,x_i)\to_= \cdots \to_=  (h_{n-1}, x_{n-1})\to_=
(h_n, x_n)= (h', x')$.
\end{itemize}
Similarly, the following are equivalent:
\begin{itemize}
\item $(h,x)\lessapprox (h',x')$;
\item either $(h,x)\lesssim (h', x')$ (if $n=0$), or else there exist $x_0, x_1, \ldots, x_i, \ldots, x_n\in Var_\Sigma$ s.t.
we have:
$(h,x)= (h_0, x_0)\leftarrow_\leq (h_1, x_1) \leftarrow_\leq (h_2, x_2)\leftarrow_\leq\cdots \leftarrow_\leq 
(h_i,x_i)\to_\leq \cdots \to_\leq (h_{n-1}, x_{n-1})\to_\leq
 (h_n, x_n)= (h', x')$.
\end{itemize}
\end{lemma}

\begin{corollary} \label{Preservation2}
If $h,h'\in H$ and $x\in Var_\Sigma$ are s.t. $h\simA h'$ and $\mathcal{L}(x)\subseteq A$, then  $(h,x)\approx(h',x)$.
\vspace{-1mm}
\end{corollary}
\vspace{-1mm}
\begin{proof}
\vspace{-2mm}
Let
$h=h_0 \ot_{A^1} h_1 \ot_{A^2} \ldots \ot_{A^i} h_{i}\to_{A^{i+1}} \ldots \to_{A^n} h_n=h'$
be the non-redundant path from $h$ to $h'$. Since $h\simA h'$, we know that $A\subseteq A_k$ for all $k$ (by Lemma \ref{path}). Using this together with $\mathcal{L}(x)\subseteq A$ (and the definition of the relation $(h,x)\to_= (h',x')$ on history-variable pairs), we obtain
$(h,x)\leftarrow_= (h_1, x) \leftarrow_=\cdots
(h_i, x)\to_= \cdots
\to_= (h_n, x)= (h', x)$.
By the Path Lemma, we have $(h,x)\approx (h', x)$.\end{proof}

\begin{corollary}\label{EasyEqual}
Let $x, x'\in Var_\Sigma$ and $h, h'\in H$, and suppose that the non-redundant path from $h$ to $h'$ is of the form
$h=h_0 \to_A h_1 \to_A \ldots \to_A h_i\to_A \ldots \to_A h_n=h'$.
Then we have $(h,x)\approx (h', x')$ iff there exists $y\in Var_\Sigma$ with $\mathcal{L}(y)\subseteq A$, $(x=y)\in last(h)$ and $(y=x')\in last(h')$.
\end{corollary}

\par\noindent\textbf{From Quasi-Model to Model} We are first defining our \emph{first-order data model} $\bD=(D, I)$ for the restricted vocabulary $\mathcal{V}_\Sigma$: as announced, the \emph{set of `values'} $D$ is the
quotient
\vspace{-2mm}
$$D\,\, :=\,\, (H\times Var_\Sigma)/\approx \, =\, \{[h,x]: (h,x)\in H\times Var_\Sigma\},\vspace{-2mm}$$
where $[h,x]$ denotes the equivalence class of $(h,x)$ modulo $\approx$, defined by
\vspace{-2mm}
$$[h,x]:=\{(h', x')\in H\times Var_\Sigma: (h,x) \approx (h', x')\}  \, \, \mbox{ (for any given pair $(h,x)\in H\times Var_\Sigma$)}.\vspace{-2mm}$$
The partial preorder $ \lessapprox$ on pairs $(h,x)\in H\times Var_\Sigma$ induces a \emph{partial order on the equivalence classes} $[h,x]\in D$, which in its turn can be extended to some \emph{total order on $D$, that we will denote by $\leq_\bD$}.\footnote{Though not unique, such a total order $\leq_\bD$ exists by the Order-Extension Principle, a well-known consequence of the Axiom of Choice.}

The \emph{interpretation function} $I$ will map $n$-ary functional symbols $F\in Funct_\Sigma$ into $n$-ary functions $I(F): D^n\to D$ given by:
$I(F) ([h,x_1], \ldots, [h,x_n]) \, :=\, [h, F(x_1, \ldots, x_n)]$ if $F(x_1, \ldots, x_n)\in Var_\Sigma$; and \linebreak
$I(f) ([h,x_1], \ldots, [h,x_n])  :=\, \bot_\bD$, otherwise; 
it will also map $n$-ary predicate symbols $P\in Pred_\Sigma\setminus \{\leq\}$ into $n$-ary relations $I(P)\subseteq D^n$ given by:
$I(P) \, := \, \{([h,x_1], \ldots, [h,x_n])\in D^n: h\in H, \ux=(x_1, \ldots, x_n)\in Var_\Sigma^n \mbox{ s.t. } P \ux\in last(h)\}$;
while the interpretation $I(\leq)$ will be just the total order $\leq_\bD$ constructed above.
\smallskip

\par\noindent\textbf{The Model} Finally, \emph{our epistemic state model} $\bM= (H, \sim, \underline{\bullet}, \underline{\bullet}(\bullet))$ is given by taking: as set of states, the set $H$ of all histories; the indistinguishability relations $\sima\subseteq H\times H$ are as defined above on histories\footnote{Note that this is meant to be a (standard) model, so the group indistinguishability relations $\simA$ are just
the intersections $\bigcap_{a\in A}\sima$, thus coinciding with the general relations $\simA$ introduced above.}; 
the valuation/assignment map $\underline{\bullet}(\bullet): H\times (V_\Sigma\cup Prop_\Sigma) \to D$ is given by putting: $\underline{h}(v) := [h,v]$ for $v\in V_\Sigma$; and $\underline{h}(p):=\top_\bD$ iff $p\in last (h)$ (-else, $\underline{h}(p):=\bot_\bD$) for $p\in Prop_\Sigma$.

\begin{lemma} \label{IntLemma} (``Interpretation Lemma'')
The interpretation $I$ is well-defined, i.e. we have the following:
\begin{enumerate}
\item $(h,\ux)\approx (h', \uy)$ implies $(h, F(\ux))\approx (h', F(\uy))$;
\item  $(h,\ux)\approx (h', \uy)$ implies that: $(P\ux\,\uz)\in last (h)$ iff $(P\uy\, \uz)\in last (h')$;
\item $I(E)$ is really the identity relation on $D$.
\end{enumerate}
\vspace{-1mm}
\end{lemma}
\vspace{-1mm}
\begin{proof}
\vspace{-2mm}
Induction on the length of the non-redundant path from $h$ to $h'$, using our conditions on types.\end{proof}

\vspace{-2mm}

\begin{lemma} \label{Pres-Value}(``Preservation of Values'')
If $h\simA h'$ and $x\in Var_\Sigma$ is s.t. $\mathcal{L}(x)\subseteq A$, then $[h,x]=[h',x]$.
\vspace{-2mm}
\end{lemma}
\vspace{-2mm}
\begin{proof}
\vspace{-1mm}
This follows immediately from Corollary \ref{Preservation2} and the definition of $[h,x]$.\end{proof}

\vspace{-1mm}

The next results use the following notation, for $h\in H$, $x\in Var_\Sigma$, $\varphi\in\Sigma$ and $A\subseteq \Agents$:
$$Val_h (x_A^\varphi)\,\, :=\,\, \{[h', x]\in D \mid h'\simA h, \varphi \in last (h')\}$$
\begin{lemma}\label{Trivial-Min}(``Trivial-Minimum Lemma'')
If $\mu_N x_A^\varphi\in Var_\Sigma$ and $h\in H$ are s.t. $Val_h (x_A^\varphi)=\emptyset$, then $[h, \mu_N x_A^\varphi]=\top_\bD$.
\vspace{-2mm}\end{lemma}
\vspace{-2mm}
\begin{proof}
\vspace{-1mm}
First, we prove an auxiliary
\emph{Claim:} $(K_A\neg\varphi)\in last(h)$. To show this,
suppose towards a contradiction that $(K_A\neg\varphi)\not\in last(h)$. This implies that $\langle K_A\rangle \varphi\in last(h)$ (given the closure conditions on types and the fact that $\mu_N x_A^\varphi\in Var_\Sigma$ implies  $\langle K_A\rangle \varphi\in\Sigma$). By the Diamond Lemma \ref{Diamond}, there exists some $h'\simA h$ with $\varphi\in last(h')$, and thus $[h', x]\in Val_h (x_A^\varphi)$. But this contradicts the assumption that $Val_h (x_A^\varphi)=\emptyset$.

Using now the above Claim, and applying condition (\ref{Trivial Min}) on types, we obtain that $(\mu_N x_A^\varphi=\top)\in last(h)$, i.e. $[h, \mu_N x_A^\varphi]=\top_\bD$, as desired.\end{proof}

\begin{lemma}\label{Trivial-Val}(``Trivial-Value Lemma'')
If $[h, \mu_N x_A^\varphi]=\top_\bD$, then $ Val_h(x_A^\varphi)\subseteq \{\top_\bD\}$. 
\vspace{-1mm}
\end{lemma} 
\vspace{-1mm}
\begin{proof}
\vspace{-2mm} 
From $[h,\mu_N x_A^\varphi]=\top_\bD$, we obtain that $(h,\mu_N x_A^\varphi)\approx (h,\top)$. By the Path Lemma \ref{PathLemma}, we must have
$(\mu_N x_A^\varphi=\top)\in last(h)$. Suppose now (towards a contradiction) that $Val_h(x_A^\varphi)\not\subseteq \{\top_\bD\}$, i.e. there exists some $h'\simA h$ with $\varphi, (x\neq\top)\in last (h')$. By Lemma \ref{equiv}, $h\simA h'$ implies that $last(h)\simA last(h')$. From this, together with the fact that $(\mu_N x_A^\varphi=\top)\in last(h)$ and that $\mathcal{L}(\mu_N x_A^\varphi=\top)=A$, =
we derive that $(\mu_N x_A^\varphi=\top)\in last(h')$ (by the definition of $\simA$ on types). 
Using condition \ref{Lower Bound} on types, we get that $(\top\leq x)\in last(h')$, which together with conditions \ref{bounds} and \ref{antisymm} on types, gives us that $(x=\top)\in last(h')$, contradicting the above assumption that $(x\neq\top)\in last (h')$.
\end{proof}

\begin{lemma} \label{Undefined-Min}(``Undefined-Minimum Lemma'')
Let $\mu_N x_A^\varphi\in Var_\Sigma$ be s.t. $[h, \mu_N x_A^\varphi]=\bot_\bD$. Then we have either $\bot_\bD\in 
Val_h (x_A^\varphi)$ or else $|Val_h (x_A^\varphi)|>N$
\vspace{-1mm}
\end{lemma} 
\vspace{-1mm}
\begin{proof}
\vspace{-2mm}
Assume towards a contradiction that $[h, \mu_N x_A^\varphi]=\bot_\bD$, but $\bot_\bD\not\in Val_h (x_A^\varphi)$; i.e.: $(\mu_N x_A^\varphi=\bot)\in last(h)$, but
$(x\neq\bot)\in last(h')$ for all $h'\simA h$ with $\varphi\in last(h')$. By Corollary \ref{KK}, $K_A^\varphi (x\neq \bot)\in last(h)$. 
This, together with $K_A (\mu_N x_A^\varphi=\bot)\in last(h)$ (which follows from $(\mu_N x_A^\varphi=\bot)\in last(h)$, by Proposition \ref{Quasi-model properties}(2), and with the closure conditions on $\Sigma$ and $\mathcal{L}(\mu_N x_A^\varphi=\bot)=A$), yield $K_A^\varphi (\mu_N x_A^\varphi \neq x)\in last(h)$. 

To show that $|Val_h (x_A^\varphi)|>N$, we will construct a sequence 
\vspace{-2mm}
$$h_0\to_A h_1\to_A \ldots \to_A h_n\to_A \ldots \to_A h_N,  \,\, \mbox{ with }  h_n \simA h \mbox{ and } \varphi\in last(h_n) \mbox{ for all $n$},
\vspace{-2mm}$$  
together with a sequence of $N$ `witnesses'
$y_1, \ldots, y_n, \ldots, y_{N}\in Var_\Sigma, \,\, \mbox{ with } \mathcal{L}(y_n)\subseteq A \mbox{ for all $n$}.$

The construction is by recursion on $n\leq N$.
For the base step ($n=0$), note that $(\langle K_A\rangle \varphi)\in last(h)$ (since otherwise we'd have $(\mu_N x_A^\varphi=\top)\in last(h)$ by condition (\ref{Trivial Min}) on types, contradicting the fact that $(\mu_N x_A^\varphi=\bot)\in last(h)$, given also condition (\ref{non-trivial}) on types). By condition (**) on quasi-models, there exists $\Delta_0\in S$ s.t. $last(h)\simA \Delta_0$ and $\varphi\in \Delta_0$. Take now $h_0: =(h, A, \Delta_0)$, which fulfills the desired specifications.

For the step $n$ of the induction (with $1\leq n\leq N$, assuming given $h_{n-1}$ and $y_1, \ldots, y_{n-1}$ with the above properties), we first take the $n^{th}$ witness $y_n$ to be any term in $Var_\Sigma$ with $\mathcal{L}(y_n)\subseteq A$ and $(x= y_n)\in last(h_{n-1})$, if such a term exists; and otherwise, we just put $y_n:=\bot$. Next, we note that, by the induction hypothesis, we have $h_{n-1}\simA h$ and $\varphi\in last(h_{n-1})$, hence $last(h_{n-1})\simA last(h)$. Together 
with the fact that $K_A^\varphi (\mu_N x_A^\varphi \neq x)\in last(h)$, this gives us that $\varphi, K_A^\varphi (\mu_N x_A^\varphi \neq x)\in last(h)$. By condition (\ref{Reaching the Min}) on types (applied to $Y:= \{y_1,\ldots, y_n\}$, where note that $\mathcal{L}(Y)\subseteq A$ and $|Y|\leq N$), we obtain 
$\langle K_A^\varphi \rangle \bigwedge_{1\leq i\leq n} (x\neq y_i)\in last(h_{n-1})$. Using again clause (**) on quasi-models, there exists some $\Delta_n\in S$ s.t.
$last(h_{n_1})\simA \Delta_n$ and $\varphi, \bigwedge_{1\leq i\leq n} (x\neq y_i)\in \Delta_n$. Take now $h_n: =(h_{n-1}, A, \Delta_n)$, which obviously fulfills the desired specifications.

Given this construction, it is clear that $[h_n, x]\in Val_h (x_A^\varphi)$ for all $0\leq n\leq N$. We will prove that \emph{all these $N+1$ values are distinct} (which immediately gives us that $|Val_h (x_A^\varphi)|>N$, as desired): 

Let $1\leq n\leq N$. It is enough to show that $[h_n, x]\neq [h_m, x]$ for all $m<n$. Suppose not, i.e. assume towards a contradiction that we have $(h_n, x)\approx (h_m, x)$ for some $m<n$. Given that the non-redundant path from $h_m$ to $h_n$ has the shape $h_m\to_A h_{m+1} \to_A \ldots \to_A h_n$, we can apply Corollary \ref{EasyEqual} to conclude that there exists some $y\in Var_\Sigma$ with $\mathcal{L}(y)\subseteq A$, $(y=x)\in last(h_m)$ and $(y=x)\in last(h_n)$. On the other hand, we also have by construction that $(x= y_{m+1})\in last(h_m)$ and $(x\neq y_{m+1})\in last (h_n)$ (since $m+1\leq n$). Putting all these together and using our conditions on equality, we obtain that $(y=y_{m+1})\in last(h_m)$ and $(y\neq y_{m+1})\in last(h_n)$. But this contradicts the fact that $last(h_m)\simA last(h_n)$ (since $h_m\simA h_n$ by construction), given that $\mathcal{L}(y=y_{m+1})\subseteq A$ (and given the definition of the relation $\simA$ on types).
\end{proof}

\begin{lemma} \label{Least-Value}(``Least-of-$N$ Values Lemma'')
If $\mu_N x_A^\varphi\in Var_\Sigma$ is s.t. $[h,\mu_N x_A^\varphi]\neq\bot_\bD, \top_\bD$, then:\begin{enumerate}
\item $(|x|_A^\varphi\leq N)\in last(h)$;
\item $|Val_h (x_A^\varphi)|\leq N$;
\item $[h,\mu_N x_A^\varphi]= min\, Val_h (x_A^\varphi)$; i.e., 
there exists some history $h_0\in H$, satisfying:
$h_0\simA h$;
$\varphi\in last(h_0)$; and
$[h, \mu_N x_A^\varphi] =  [h_0, x] \leq_\bD [h', x]$ for all $[h',x]\in Val_h (x_A^\varphi)$.
\end{enumerate}
\vspace{-2mm}
\end{lemma}
\vspace{-2mm}
\begin{proof}\vspace{-2mm}
For part 1, from $[h,\mu_N x_A^\varphi]\neq\bot_\bD$ we obtain that $(\mu_N x_A^\varphi\neq \bot)\in last (h)$, so by condition (\ref{Undefined Min}) on types we have $(|x|_A^\varphi\leq N)\in last(h)$, as desired.

For part 2, we use part 1 and Lemma \ref{K}, to obtain that $(\bigvee_{1\leq i\leq n} x=i_N.x_A^\varphi)\in last(h')$ holds for all $[h',x]\in Val_h (x_A^\varphi)$; hence, for every $[h',x]\in Val_h (x_A^\varphi)$ there exists $i\in \{1, \ldots, N\}$ s.t. $[h', x]=[h', i_N.x_A^\varphi]= [h, i_N.x_A^\varphi]$ (where the last equality follows by the Corollary \ref{Preservation2} from the
fact that $\mathcal{L}(i_N. x_A^\varphi)=A$ together with $h'\simA h$).
Thus, we have that $Val_h (x_A^\varphi) \subseteq \{ [h, i_N.x_A^\varphi] \mid 1\leq i\leq N\}$, which implies that $|Val_h (x_A^\varphi)|\leq N$.

For part 3, we use the assumption that $[h,\mu_N x_A^\varphi]\neq\top_\bD$ and Lemma \ref{Trivial-Min} to obtain $Val_h (x_A^\varphi)\neq \emptyset$; i.e., there exists $h_1\simA h$ s.t. $\varphi \in last(h_1)$. Putting this together with the fact that $(|x|_A^\varphi\leq N)\in last(h_1)$ (which follows from part 1 and $h_1\simA h$, using $\mathcal{L}(|x|_A^\varphi\leq N)=A$ and Lemma \ref{equiv}), and applying condition (\ref{Reaching the Min}) on types (with $Y:= Var^N (x_A^\varphi)=\{i_N. x_A^\varphi : 1\leq i\leq N\}$, while recalling that by definition $|x|_A^\varphi\leq N := K_A^\varphi (x\in  Var^N (x_A^\varphi))$, we conclude that $\langle K_A^\varphi \rangle (x=\mu_N x_A^\varphi)\in last(h_1)$. By the Diamond Lemma \ref{Diamond}, there exists some $h_0\simA h_1\simA h$ with $\varphi, (x=\mu_N x_A^\varphi)\in last(h_0)$, hence $[h_0, x]=[h_0, \mu_N x_A^\varphi] = [h, \mu_N x_A^\varphi]$ (with the last equality due to $h_0\simA h$ and $\mathcal{L}(\mu_N x_A^\varphi)=A$, by Lemma \ref{Pres-Value}). Finally, to prove that $[h,\mu_N x_A^\varphi]= min\, Val_h (x_A^\varphi)$, let  $[h',x]\in Val_h (x_A^\varphi)$, i.e., $h'\simA h$ with $\varphi\in last(h)$, and we need to show that $[h, \mu_N x_A^\varphi]\leq_\bD [h', x]$.  By condition (\ref{Lower Bound}) on types, we have $(\mu_N x_A^\varphi\leq x)\in last(h')$, hence
$[h', \mu_N x_A^\varphi]\leq_\bD [h', x]$. Since $h\simA h'$ and $\mathcal{L}(\mu_N x_A^\varphi)=A$, Lemma \ref{Pres-Value} gives us that $[h, \mu_N x_A^\varphi]= [h', \mu_N x_A^\varphi]\leq_\bD [h', x]$, as desired.\end{proof}

\begin{lemma} \label{Truth}(``Truth Lemma'')
For every 
$\varphi\in \Sigma$ and 
$x\in Var_\Sigma$, the following hold for all 
$h\in H$:
\begin{enumerate}
\item $h\models_\bM \varphi$ iff $\varphi\in last(h)$;
\item $\underline{h}(x)_\bM=[h,x]$.
\end{enumerate}
\end{lemma}
\begin{proof}
\vspace{-3mm}
We prove 1.\&2. \textit{for all expressions $\alpha\in \Sigma\cup Var_\Sigma$ by induction on sub-expression complexity}.

(i) \emph{Predicative Atoms}: $\varphi=P\ux$, with $\ux=(x_1, \ldots, x_n)$.
For $P\in Pred\setminus \{\leq\}$, we have the equivalences:
\noindent $h \models P\ux$ iff $h(\ux)\in I(P)$ iff (by the induction hypothesis for claim (2)) $([h,x_1], \ldots, [h, x_n])\in I(P)$ iff $\exists h'\in H \mbox{ s.t. }  (\, (h,\ux)\approx (h', \ux) \& (P \ux)\in last(h') \, )$ iff
$(P \ux)\in last(h)$ (by Lemma \ref{IntLemma}).
For $\leq$, we have the equivalencies:
\noindent $h \models x\leq y$ iff $(h,x) \lesssim (h,y)$ iff $(h,x) \lessapprox (h,y)$ iff $[h,x]\leq_\bD [h,y]$ (by the Path Lemma).

(ii) \emph{Propositional atoms}: $p\in Prop_\Sigma$. By definition, $h\models p$ iff $\underline{h}(p)= \top_\bD$ iff $p\in last(h)$.

(iii) \emph{Basic variables}: $v\in V_\Sigma$. By definition, $h(v)_\bM=\underline{h} (v)=[h,v]$.

(iv) \emph{Boolean Case}: $\varphi\to\psi$. These is trivial, using conditions \ref{neg} and \ref{conj} on types.

(v) \emph{$K_A$-modal Case}: $(K_A \varphi)\in \Sigma$. From \emph{left-to-right}: assume (towards a contradiction) that $h\models K_A\varphi$ but $(K_A\varphi)\not\in last(h)$. By condition \ref{neg} on types, we have $(\langle K_A\rangle\!\sim\!\varphi)\in last(h)$, and so by the property (**) of quasi-models, there exists a type $\Delta'\in S$ s.t. $last(h)\simA \Delta'$ and $(\sim\!\varphi)\in\Delta'$. Take $h':=(h, A, \Delta')$: this is a well-defined history in $H$, satisfying $h \simA h'$ and $last(h')=\Delta'$. From $h\models K_A\varphi$ we infer that $h'\models \varphi$, and so by the induction hypothesis we have $\varphi\in last(h')=\Delta'$, in contradiction to $(\sim\varphi)\in\Delta'$.

For the \emph{right-to-left} direction: we assume that $(K_A\varphi)\in last (h)$, and we have to prove that $h\models K_A\varphi$. For this, let $h'\in H$ be s.t. $h\simA h'$, and we need to show that $h'\models \varphi$. From $h\simA h'$ we obtain that $last(h)\simA last(h')$ (by Lemma \ref{equiv}), which together with $h\models K_A\varphi$ gives us $\varphi\in last(h')$ (by Proposition \ref{Access2}). Applying the induction hypothesis, we conclude that $h'\models \varphi$, as desired.

(vi) \emph{Terms defined by cases} $\varphi\!\!\to\!\!x|y$. Given $(\varphi\!\!\to\!\!x|y)\in\Sigma$, we have $\varphi\in\Sigma$, so we have that either $\varphi \in last(h)$, or else $(\sim\varphi)\in last(h)$. In the first case, we have $(\varphi\!\!\to\!\!x|y=x)\in last(h)$ (by condition \ref{cases} on types), hence $[h, \varphi\!\!\to\!\!x|y]=[h,x]$; and on the other hand, $\varphi \in last(h)$ yields $h\models\varphi$ (by the induction hypothesis for Claim 1), so by definition we have $\underline{h}(\varphi\!\!\to\!\!x|y)_\bM=\underline{h}(x)_\bM$, and by the induction hypothesis for Claim 2 we have $\underline{h}(x)_\bM=[h,x]$, thus obtaining
$\underline{h}(\varphi\!\!\to\!\!x|y)_\bM= \underline{h}(x)_\bM= [h,x]=[h, \varphi\!\!\to\!\!x|y]$, as desired. The second case is similar: from $(\sim\varphi)\in last(h)$ we get
$(\varphi\!\!\to\!\!x|y=y)\in last(h)$ (by condition \ref{cases} on types), hence $[h, x|_\varphi]=[h,y]$; and on the other hand, $(\sim\varphi) \in last(h)$ implies $\varphi\not\in last(h)$ (by the consistency of types), which by the induction hypothesis for Claim 1 yields $h\not\models\varphi$, so by definition we have $\underline{h}(\varphi\!\!\to\!\!x|y)_\bM=\underline{h}(y)_\bM$, and by the induction hypothesis for Claim 2 we have $\underline{h}(y)_\bM=[h,y]$, thus obtaining $\underline{h}(\varphi\!\!\to\!\!x|y)_\bM= \underline{h}(y)_\bM= [h,y]=[h,\varphi\!\!\to\!\!x|y]$, as desired.

(vii) \emph{Functional terms} $F(\ux)\in Var_{\Sigma}$ for some $\ux=(x_1,\ldots, x_n)$. We assume by the induction hypothesis that $\underline{h}(\ux)_\bM=[h,\ux]$, and using the semantics of functional terms and the definition of $I(F)$ in our history model, we obtain
$\underline{h}(F(\ux))_\bM= (I(F)) (h(\ux))= (I(F)) [h,\ux]= [h, F(\ux)]$, as desired.

(viii)  \emph{Least-value terms} $\mu_N x_A^\varphi\in Var_{\Sigma}$. 
Using the notation
$Val_h(x_A^\varphi):=\{[h', x] \in D : h'\simA h, \varphi\in last(h')\}$ and the induction hypothesis for Claims (1) and (2), we have 
$Val_h(x_A^\varphi) = Val(x_A^\varphi)_{h,\bM}$,
where
$Val(x_A^\varphi)_{h,\bM}:=\{\underline{h'}(x)_\bM : h'\simA h, h'\models_\bM \varphi\}$ is the notation from Section \ref{Logic}.
We distinguish \emph{three subcases}:

\emph{Case 1}: $\,[h,\mu_N x_A^\varphi]\,=\,\top_\bD$.  By the Trivial-Value Lemma \ref{Trivial-Val}, we have  $\,Val_h(x_A^\varphi)\, \subseteq \,  \{\top_\bD\,\}$, and hence \linebreak $min_N Val_h(x_A^\varphi)\,=\,\top_\bD$ (since by definition we have $\, min_N \emptyset\, =\, \min_N \{\top_\bD\}\, =\, \top_\bD$). Given this, we obtain that $\, \underline{h}(\mu_N x_A^\varphi)_\bM \, =\,
min_N Val(x_A^\varphi)_{h,\bM} = min_N Val_h(x_A^\varphi)=\bot_\bD= [h,\mu x_A^\varphi]$, as desired.

\emph{Case 2}: $[h,\mu_N x_A^\varphi]=\bot_\bD$. By the Undefined-Minimum Lemma \ref{Undefined-Min}, we have either 
$\bot_\bD\in Val_h (x_A^\varphi)$ or else $|Val_h (x_A^\varphi)|>N$. In both cases, we get $min_N Val_h (x_A^\varphi)\, = \, \bot_\bD$ by definition, so we obtain \linebreak 
$\underline{h}(\mu_N x_A^\varphi)_\bM =
min_N Val(x_A^\varphi)_{h,\bM}= min_N Val_h (x_A^\varphi)=\bot_\bD= [h,\mu_N x_A^\varphi]$, as desired.

\emph{Case 3}: $[h,\mu_N x_A^\varphi]\neq \top_\bD, \bot_\bD$. By Lemma \ref{Least-Value}, we have both
$|Val_h (x_A^\varphi)|\leq N$ and $[h,\mu_N x_A^\varphi]=min\, Val_h (x_A^\varphi)$, i.e.,  $[h,\mu_N x_A^\varphi]=min_N Val_h (x_A^\varphi)$, thus
$\underline{h}(\mu_N x_A^\varphi)_\bM =min_N Val(x_A^\varphi)_{h,\bM}= min_N Val_h (x_A^\varphi)=[h,\mu_N x_A^\varphi]$, as desired.
\end{proof}


%

\smallskip

\par\noindent\emph{\textbf{Proof of  Proposition \ref{Satisfiability}}}: If $S$ is a quasi-model for $\varphi_0$, then by applying claim 1 of the Truth Lemma \ref{Truth} to $\varphi_0$ and to the history $h_0:=(\Delta_0)$, we conclude that $h_0\models \varphi_0$ in our history-based model $\bM$ above.

\vspace{-3mm}

\section{Completeness and Decidability Proofs for $LDAE$}\label{B}
\vspace{-2mm}
In this section we prove Theorem \ref{CompDLKV}. We first establish our co-expressivity result: \emph{$LDAE$ and $LDA$ are provably co-expressive}. For this, we need a few preliminary notions and results.

\smallskip\par\noindent\textbf{Reducible expressions}
A formula $\theta$ in $LDAE$ is said to be \emph{reducible} if it is provably equivalent to a 'static' formula; i.e., if there exists some formula $\theta'$ in the static fragment $LDA$ s.t.
$\vdash \, \theta \leftrightarrow \theta'$
is provable in $\mathbf{LDAE}$. A term $x\in Var_{LDAE}$ is \emph{reducible} if it is provably equal to a static term; i.e., if there exists some $x'\in Var_{LDA}$ s.t. $\vdash \, x=x'$ is provable in $\mathbf{LDAE}$. Finally, an event $\be=(\bE,e)$ is \emph{reducible} if all preconditions $pre_f$ and all post-conditions of the form $\underline{e}(p)$ or $\underline{e}(v)$ are reducible (for all $f\in E$, $p\in Prop$ and $v\in V$).

\begin{lemma}\label{re} \emph{(Replacement of Equivalents)}
Suppose that $\vdash \varphi \leftrightarrow \varphi'$, $\vdash \theta \leftrightarrow \theta'$, $\vdash x=x'$, $\vdash y=y'$ and $\vdash x_i =x'_i$ (for all $i=1,n$) are provable in $\mathbf{LDAE}$. Then the following are also provable in $\mathbf{LDAE}$:\\

\begin{minipage}{0.45\textwidth}
\begin{enumerate}
\item $\vdash \varphi\!\!\to\!\!x|y = \varphi\!\!\to\!\!x'|y'$
\item $\vdash F(x_1, \ldots, x_n)= F(x'_1, \ldots, x'_n)$
\item $\vdash \mu_N x_A^\varphi = \mu_N (x')_A^{\varphi'}$
\item $\vdash Px_1 \ldots x_n \leftrightarrow Px'_1 \ldots x'_n$
   \end{enumerate}
   \end{minipage}
\hfill   
\begin{minipage}{0.45\textwidth}
\begin{enumerate}
\setcounter{enumi}{4}
\item $\vdash (\varphi \to \theta) \leftrightarrow (\varphi' \to \theta')$
\item $\vdash K_A \varphi \leftrightarrow K_A\varphi'$
\item $\vdash [\be]\varphi \leftrightarrow [\be]\varphi'$
\item $\vdash \be(x)=\be(x')$
    \end{enumerate}
    \end{minipage}
\end{lemma}

\smallskip

We first prove a preliminary ``one-step reduction'' result:

\begin{lemma}\label{one-step reduction} Let $\be$ be any reducible event.
Then, for every 'static' formula $\theta$ in $LDA$ and every 'static' term $x\in Var_{LDA}$, the formula $[\be] \theta$ and the term $\be(x)$ are also reducible.
\end{lemma}
\vspace{-2mm}
\begin{proof}
\vspace{-1mm}
Since $\be=(\bE, e)$ is reducible, its precondition $pre_e$ is also reducible. Let us fix some static formula $\rho$ s.t. $\vdash pre_e\leftrightarrow \rho$ is provable in $\mathbf{LDAE}$. 
We'll prove both claims simultaneously, by induction on sub-expression complexity (-and we'll make liberal use of Lemma \ref{re}, without explicitly mentioning it):

\smallskip

For $\theta:=p$:  since $e$ is reducible, there exists some static formula $\theta'$ s.t. $\vdash \underline{e}(p)\leftrightarrow \theta'$. 
Using the Change of Facts axiom, we can see that $[\be]p$ is provably equivalent to $\rho \to \theta'$.

For $\theta:=Px_1\ldots x_n$: by the induction hypothesis, for all $i=1,n$ we have $\vdash \be(x_i)=x'_i$ for some static terms $x'_1, \ldots, x'_n\in Var_{LDA}$. Using also the Indiscernability axiom and the Property Change reduction axiom, $[\be] \theta$ is provably equivalent to $pre_e \to P x'_1\ldots x'_n$, and hence also to $\rho\to P x'_1\ldots x'_n$.

For $\theta \, :=\, \varphi\to \psi$ is similar: by induction, there exist static formulas $\varphi_e$ and $\psi_e$ s.t. $\vdash [\be] \varphi \leftrightarrow \varphi'$  and $\vdash [\be] \psi \leftrightarrow \psi'$ are provable. Using the $[\be]$-Distributivity axiom, we obtain $\vdash [\be] \theta \leftrightarrow (\varphi' \to \psi')$.

For $\theta:= K_A\varphi$: by induction, for each $f\in E$ there exists some static formula $\varphi_f$ s.t. $\vdash [\mathbf{f}] \varphi \leftrightarrow \varphi_f$. Putting this together with the Knowledge Update axiom, we get $\vdash [\be] \theta \leftrightarrow \bigwedge \{\rho \to K_{\underline{e}(A)} \varphi_f : f\simA e\})$.

For $x:=v$ (basic variable): since $e$ is reducible, there exists some static term $x'$ s.t. $\vdash \underline{e}(v)= x'$. Using this and the Value Change axiom, we obtain that $\vdash \be(x)\, =\,(\rho\!\!\to\!\! x'|\top)$, so $\be(x)$ is reducible.

For $x:=F(x_1, \ldots, x_n)$: by the induction hypothesis, there exist static terms $x'_1, \ldots, x'_n$ s.t. $\vdash \be(x_i)=x'_i$ for all $i\leq n$. Using the Functional Change reduction axiom, we obtain $\vdash \be(x)\, =\, (\rho\!\!\to\!\! F(x'_1, \ldots, x'_n))$.

For $x:=\varphi\!\!\to\!\!y|z$: by the induction hypothesis, there exist static terms $y', z'$ and static formula $\varphi'$, s.t.
$\vdash \be(y)=y'$, $\vdash \be(z)=z'$ and $\vdash [\be]\varphi \leftrightarrow \varphi'$ are provable. Using these, as well as the Case Change reduction axiom, we obtain $\vdash \be(x)\, =\, (\varphi'\!\!\to\!\! y'|z')$, and so $\be(x)$ is reducible.

For $x:=\mu_N y_A^\varphi$: first note that by definition, the reducibility of $\be=(\bE, e)$ implies the reducibility of \emph{all} $\mathbf{f}=(\bE, f)$ with $f\in E$.
By the induction hypothesis and Proposition \ref{theorems2}, for every $f\in E$ there exist a static formula $\varphi_f$ and a static term $y_f$, s.t. $\vdash \langle \boldf\rangle \varphi \leftrightarrow \varphi_f$ and $\vdash \boldf(y)=y_f$.
Thus, if we put
$$\eta \, \,\, :=\,\,\, \{n_N (y_f)_{\underline{f}(A)}^{\varphi_f} : f\simA e, n\leq N\}^\Diamond\leq N,$$ then $\vdash \eta \leftrightarrow Val_e^N (x_A\varphi)^\Diamond \leq N$ is provable in $\mathbf{LDAE}$. Using the Min. Val. Change axiom, $\be(\mu_N y_A^\varphi)$ is provably equal to 
$\rho\!\!\to\!\! \left(\eta\!\!\to\!\!
min\, \{\mu_N (y_f)_{\underline{f}(A)}^{\varphi_f}: f\simA e\}|\bot \right)$.
\end{proof}
Using this, we can establish our full reduction result:

\begin{lemma}\label{reduction} \emph{(Co-expressivity of $LDAE$ and $LDA$)}
All expressions (formulas, terms and events) $\alpha$ of $LDAE$ are reducible. As a consequence, $LDAE$ and $LDA$ have the same expressive power.
\end{lemma}
\vspace{-2mm}
\begin{proof}
\vspace{-1mm}
Induction on the subexpression-complexity of the expression $\alpha$.
The base cases $\alpha:= p\in Prop$ and $\alpha:=v\in V$ are trivial.
The cases $\alpha:= \varphi\!\!\to\!\!x|y$, $\alpha:= F(\ux)$, $\alpha:=\mu_N x_A^\varphi$, $\alpha:= P\ux$, $\alpha:=\varphi\to \theta$ and $\alpha:=K_A\varphi$ are straightforward: use the induction hypothesis and parts 1-6 of Lemma \ref{re}.

For $\alpha:=[\be]\varphi$: by induction, $\be$ and $\varphi$ are reducible, so there exists static $\varphi'$ s.t. $\vdash \varphi \leftrightarrow \varphi'$. By Lemma \ref{re}.7,  $\vdash \alpha \leftrightarrow [\be]\varphi'$, and by Lemma \ref{one-step reduction} $[\be]\varphi'$ is reducible (since $\be$ is reducible and $\varphi'$ is static), so $\vdash [\be]\varphi' \leftrightarrow \varphi''$ for some static $\varphi''$. Using transitivity of equivalence, we get $\vdash \alpha \leftrightarrow \varphi''$, so $\alpha$ is reducible.

For $\alpha:= \be(x)$: by induction, $\be$ and $x$ are reducible, so there exists a static term $x'$ s.t. $\vdash x=x'$. By Lemma \ref{re}.8, we have $\vdash \alpha=\be(x')$, and by Lemma \ref{one-step reduction} $\be(x')$ is reducible (since $\be$ is reducible and $x'$ is static), so $\vdash \be(x')=x''$ for some static term $x''$. Using transitivity of equality, we obtain $\vdash \alpha = x''$.

Finally, the case $\alpha:=\be=(\bE,e)$ is straightforward: for all $f\in E$, $p\in Prop$ and $v\in V$, the preconditions $pre_f$ and postconditions $\underline{f}(p)$ and $\underline{f}(v)$ are $<\be$ in the subexpression complexity order, so they are all reducible (by the induction hypothesis), and hence $\be$ is also reducible (by definition).
\end{proof}

\vspace{-1mm}
\par\noindent\emph{\textbf{Proof of
Theorem \ref{CompDLKV}}}: Completeness follows immediately from Lemma \ref{reduction}, the soundness of $\mathbf{LDAE}$, and the completeness of the system $\mathbf{LDA}$ (Theorem \ref{CompLKV}); decidability follows from the decidability of the static logic $LDA$ (Theorem \ref{CompLKV}) and the co-expressivity of $LDA$ and $LDAE$ (Lemma \ref{reduction}).
\vspace{-3mm}

\section{Conclusions and Future Work}
\vspace{-2mm}
In this paper we axiomatize a decidable but richly expressive logic, that deals with a large class of data-exchange events, while capturing group knowledge of both propositional and non-propositional data, as well as a group's ability to narrow down the values of a variable to finitely many possibilities. 

There are a number of things still left to do. 
As already mentioned, one
can add common knowledge operators $C_A\varphi$, but pre-encoding their dynamics requires a generalization to polyadic conditionals $C_A^e \varphi$ as in \cite{BS24}. 
We leave this for a journal version, where we also plan to explore the relationships of our formalism with the Logic of Functional Dependence \cite{BvB2020a} and Graded (Multi)Modal Logic.

Our axioms are somewhat complicated due to the fact that we did \emph{not} succeed to prove FMP (Finite Model Property) for our logics. 
This would have allowed us to replace cutoff-minimization $\mu_N x_A^\varphi$ by 
plain minimization $min\, x_A^\varphi$ 
(over all $A$-possible $x$-values given $\varphi$), which would greatly simplify our axioms. 
On the other hand, we \emph{don't} have a counterexample to FMP, so this issue is still an open question.
\vspace{-4mm}
\bibliographystyle{eptcs}
\bibliography{AiML-REFS}

\begin{thebibliography}{10}
\providecommand{\bibitemdeclare}[2]{}
\providecommand{\surnamestart}{}
\providecommand{\surnameend}{}
\providecommand{\urlprefix}{Available at }
\providecommand{\url}[1]{\texttt{#1}}
\providecommand{\href}[2]{\texttt{#2}}
\providecommand{\urlalt}[2]{\href{#1}{#2}}
\providecommand{\doi}[1]{doi:\urlalt{https://doi.org/#1}{#1}}
\providecommand{\eprint}[1]{arXiv:\urlalt{https://arxiv.org/abs/#1}{#1}}
\providecommand{\bibinfo}[2]{#2}

\bibitemdeclare{article}{AgotnesWang2017}
\bibitem{AgotnesWang2017}
\bibinfo{author}{Thomas \surnamestart Agotnes\surnameend} \&
  \bibinfo{author}{Yi~N. \surnamestart Wang\surnameend} (\bibinfo{year}{2017}):
  \emph{\bibinfo{title}{Resolving Distributed Knowledge}}.
\newblock {\slshape \bibinfo{journal}{Artificial Intelligence}}
  \bibinfo{volume}{252}, pp. \bibinfo{pages}{1--21},
  \doi{10.1016/j.artint.2017.07.002}.

\bibitemdeclare{incollection}{Baltag2010}
\bibitem{Baltag2010}
\bibinfo{author}{Alexandru \surnamestart Baltag\surnameend}
  (\bibinfo{year}{2010}): \emph{\bibinfo{title}{Presentation: ``What is DEL
  good for?''}}.
\newblock In: {\slshape \bibinfo{booktitle}{ESSLI Workshop on `Logic,
  Rationality and Intelligent Interaction' (organized by J. van Benthem and E.
  Pacuit)}}.
\newblock
  \urlprefix\url{https://cs.stanford.edu/~epacuit/lograt/wkshp-esslli2010.html}.

\bibitemdeclare{incollection}{Baltag2016}
\bibitem{Baltag2016}
\bibinfo{author}{Alexandru \surnamestart Baltag\surnameend}
  (\bibinfo{year}{2016}): \emph{\bibinfo{title}{To Know is to Know the Value of
  a Variable}}.
\newblock In: {\slshape \bibinfo{booktitle}{Advances in Modal Logic 2016}},
  \bibinfo{volume}{11}, \bibinfo{publisher}{College Publications}, pp.
  \bibinfo{pages}{135--155}.
\newblock \bibinfo{note}{ISBN:978-1-84890-201-5}.

\bibitemdeclare{incollection}{BBS2016}
\bibitem{BBS2016}
\bibinfo{author}{Alexandru \surnamestart Baltag\surnameend},
  \bibinfo{author}{Rachel \surnamestart Boddy\surnameend} \&
  \bibinfo{author}{Sonja \surnamestart Smets\surnameend}
  (\bibinfo{year}{2018}): \emph{\bibinfo{title}{Group Knowledge in
  Interrogative Epistemology}}.
\newblock In: {\slshape \bibinfo{booktitle}{Outstanding Contributions to
  Logic}}, \bibinfo{volume}{12}, \bibinfo{publisher}{Springer}, pp.
  \bibinfo{pages}{131--164}, \doi{10.1007/978-3-319-62864-6_5}.

\bibitemdeclare{incollection}{BMS}
\bibitem{BMS}
\bibinfo{author}{Alexandru \surnamestart Baltag\surnameend},
  \bibinfo{author}{Lawrence \surnamestart Moss\surnameend} \&
  \bibinfo{author}{Slawomir \surnamestart Solecki\surnameend}
  (\bibinfo{year}{1998}): \emph{\bibinfo{title}{The Logic of Public
  Announcements, Common Knowledge, and Private Suspicions}}.
\newblock In: {\slshape \bibinfo{booktitle}{Proceedings TARK 98}},
  \bibinfo{publisher}{Morgan Kaufmann}, pp. \bibinfo{pages}{43--56}.

\bibitemdeclare{article}{NEW-BMS}
\bibitem{NEW-BMS}
\bibinfo{author}{Alexandru \surnamestart Baltag\surnameend},
  \bibinfo{author}{Lawrence \surnamestart Moss\surnameend} \&
  \bibinfo{author}{Slawomir \surnamestart Solecki\surnameend}
  (\bibinfo{year}{2023}): \emph{\bibinfo{title}{Logics for epistemic actions:
  completeness, decidability, expressivity}}.
\newblock {\slshape \bibinfo{journal}{Logics}} \bibinfo{volume}{1(2)}, pp.
  \bibinfo{pages}{97--147}, \doi{10.3390/logics1020006}.

\bibitemdeclare{article}{BM}
\bibitem{BM}
\bibinfo{author}{Alexandru \surnamestart Baltag\surnameend} \&
  \bibinfo{author}{Lawrence~S. \surnamestart Moss\surnameend}
  (\bibinfo{year}{2004}): \emph{\bibinfo{title}{Logics for Epistemic
  Programs}}.
\newblock {\slshape \bibinfo{journal}{Synthese}} \bibinfo{volume}{139(2)}, pp.
  \bibinfo{pages}{165--224}, \doi{10.1023/B:SYNT.0000024912.56773.5e}.

\bibitemdeclare{incollection}{BMD}
\bibitem{BMD}
\bibinfo{author}{Alexandru \surnamestart Baltag\surnameend},
  \bibinfo{author}{Lawrence~S. \surnamestart Moss\surnameend} \&
  \bibinfo{author}{Hans \surnamestart van Ditmarsch\surnameend}
  (\bibinfo{year}{2008}): \emph{\bibinfo{title}{Epistemic Logic and Information
  Update}}.
\newblock In \bibinfo{editor}{P.~\surnamestart Adriaans\surnameend} \&
  \bibinfo{editor}{J.~\surnamestart {van Benthem}\surnameend}, editors:
  {\slshape \bibinfo{booktitle}{Philosophy of Information, part of series:
  Handbook of the Philosophy of Science}}, \bibinfo{volume}{8},
  \bibinfo{publisher}{Elsevier}, pp. \bibinfo{pages}{361--465},
  \doi{10.1016/C2009-0-16481-4}.

\bibitemdeclare{incollection}{sep-dynamic-epistemic}
\bibitem{sep-dynamic-epistemic}
\bibinfo{author}{Alexandru \surnamestart Baltag\surnameend} \&
  \bibinfo{author}{Bryan \surnamestart Renne\surnameend}
  (\bibinfo{year}{2016}): \emph{\bibinfo{title}{{Dynamic Epistemic Logic}}}.
\newblock In \bibinfo{editor}{Edward~N. \surnamestart Zalta\surnameend},
  editor: {\slshape \bibinfo{booktitle}{The {Stanford} Encyclopedia of
  Philosophy}}, \bibinfo{edition}{{W}inter 2016} edition,
  \bibinfo{publisher}{Metaphysics Research Lab, Stanford University}.
\newblock
  \urlprefix\url{https://plato.stanford.edu/archives/win2016/entries/dynamic-epistemic/}.

\bibitemdeclare{incollection}{BS20}
\bibitem{BS20}
\bibinfo{author}{Alexandru \surnamestart Baltag\surnameend} \&
  \bibinfo{author}{Sonja \surnamestart Smets\surnameend}
  (\bibinfo{year}{2020}): \emph{\bibinfo{title}{Learning what Others Know}}.
\newblock In \bibinfo{editor}{L.~\surnamestart Kovacs\surnameend} \&
  \bibinfo{editor}{E.~\surnamestart Albert\surnameend}, editors: {\slshape
  \bibinfo{booktitle}{LPAR23 proceedings, EPiC Series in Computing}},
  \bibinfo{volume}{73}, pp. \bibinfo{pages}{90--110},
  \doi{10.48550/arXiv.2109.07255}.

\bibitemdeclare{inproceedings}{BS24}
\bibitem{BS24}
\bibinfo{author}{Alexandru \surnamestart Baltag\surnameend} \&
  \bibinfo{author}{Sonja \surnamestart Smets\surnameend}
  (\bibinfo{year}{2024}): \emph{\bibinfo{title}{Logics for Data Exchange and
  Communication}}.
\newblock In: {\slshape \bibinfo{booktitle}{Advances in Modal Logic 2024}},
  \bibinfo{volume}{15}, \bibinfo{publisher}{College Publications}, pp.
  \bibinfo{pages}{147--169}.
\newblock \bibinfo{note}{ISBN:978-1-84890-467-5}.

\bibitemdeclare{incollection}{BS25}
\bibitem{BS25}
\bibinfo{author}{Alexandru \surnamestart Baltag\surnameend} \&
  \bibinfo{author}{Sonja \surnamestart Smets\surnameend}
  (\bibinfo{year}{2025}): \emph{\bibinfo{title}{Group Knowledge of Hypothetical
  Values}}.
\newblock In: {\slshape \bibinfo{booktitle}{Electronic Proceedings in
  Theoretical Computer Science (TARK 2025)}}, \bibinfo{volume}{437}, pp.
  \bibinfo{pages}{135--154}, \doi{10.4204/EPTCS.437.15}.

\bibitemdeclare{article}{BvB2020a}
\bibitem{BvB2020a}
\bibinfo{author}{Alexandru \surnamestart Baltag\surnameend} \&
  \bibinfo{author}{Johan \surnamestart {van Benthem}\surnameend}
  (\bibinfo{year}{2021}): \emph{\bibinfo{title}{A Simple Logic of Functional
  Dependence}}.
\newblock {\slshape \bibinfo{journal}{Journal of Philosophical logic}}
  \bibinfo{volume}{50}, pp. \bibinfo{pages}{939--1005},
  \doi{10.1007/s10992-020-09588-z}.

\bibitemdeclare{inproceedings}{BenthemMinica2}
\bibitem{BenthemMinica2}
\bibinfo{author}{Johan \surnamestart van Benthem\surnameend} \&
  \bibinfo{author}{{\c{S}}tefan \surnamestart Minic{\u{a}}\surnameend}
  (\bibinfo{year}{2009}): \emph{\bibinfo{title}{Toward a Dynamic Logic of
  Questions}}.
\newblock In: {\slshape \bibinfo{booktitle}{LORI 2009 Proceedings, Lecture
  Notes in Computer Science}}, \bibinfo{volume}{5834},
  \bibinfo{publisher}{Springer}, pp. \bibinfo{pages}{27--41},
  \doi{10.1007/978-3-642-04893-7_3}.

\bibitemdeclare{mastersthesis}{Boddy2014}
\bibitem{Boddy2014}
\bibinfo{author}{Rachel \surnamestart Boddy\surnameend} (\bibinfo{year}{2014}):
  \emph{\bibinfo{title}{Epistemic Issues and Group Knowledge}}.
\newblock \bibinfo{type}{Master of logic thesis, mol-2014-03},
  \bibinfo{school}{University of Amsterdam}.
\newblock \urlprefix\url{https://eprints.illc.uva.nl/id/eprint/921/}.

\bibitemdeclare{article}{Ding}
\bibitem{Ding}
\bibinfo{author}{Yifeng \surnamestart Ding\surnameend} (\bibinfo{year}{2016}):
  \emph{\bibinfo{title}{The axiomatization and complexity of Knowing-What-Logic
  on model class K, Epistemic Logic with Functional Dependency Operator}}.
\newblock {\slshape \bibinfo{journal}{arXiv 1609.07684}},
  \doi{10.48550/arXiv.1609.07684}.

\bibitemdeclare{book}{DHK}
\bibitem{DHK}
\bibinfo{author}{Hans \surnamestart van Ditmarsch\surnameend},
  \bibinfo{author}{Wiebe \surnamestart van~der Hoek\surnameend} \&
  \bibinfo{author}{Barteld \surnamestart Kooi\surnameend}
  (\bibinfo{year}{2008}): \emph{\bibinfo{title}{Dynamic Epistemic Logic}}.
\newblock \bibinfo{volume}{337}, \bibinfo{publisher}{Springer},
  \doi{10.1007/978-1-4020-5839-4}.

\bibitemdeclare{incollection}{vEGWang}
\bibitem{vEGWang}
\bibinfo{author}{Jan \surnamestart van Eijck\surnameend},
  \bibinfo{author}{Malvin \surnamestart Gattinger\surnameend} \&
  \bibinfo{author}{Yanjing \surnamestart Wang\surnameend}
  (\bibinfo{year}{2017}): \emph{\bibinfo{title}{Knowing Values and Public
  Inspection}}.
\newblock In: {\slshape \bibinfo{booktitle}{Lecture Notes in Computer
  Science}}, \bibinfo{volume}{10119}, \bibinfo{publisher}{Springer}, pp.
  \bibinfo{pages}{77--90}, \doi{10.1007/978-3-662-54069-5_7}.

\bibitemdeclare{mastersthesis}{Goldbach2015}
\bibitem{Goldbach2015}
\bibinfo{author}{Roosmarijn \surnamestart Goldbach\surnameend}
  (\bibinfo{year}{2015}): \emph{\bibinfo{title}{Modelling Democratic
  Deliberation}}.
\newblock \bibinfo{type}{Master of logic thesis, mol-2015-05},
  \bibinfo{school}{University of Amsterdam}.
\newblock \urlprefix\url{https://eprints.illc.uva.nl/id/eprint/946/}.

\bibitemdeclare{incollection}{GuWang}
\bibitem{GuWang}
\bibinfo{author}{Tao \surnamestart Gu\surnameend} \& \bibinfo{author}{Yanjing
  \surnamestart Wang\surnameend} (\bibinfo{year}{2016}):
  \emph{\bibinfo{title}{``Knowing value'' logic as a normal modal logic}}.
\newblock In: {\slshape \bibinfo{booktitle}{Advances in Modal Logic}},
  \bibinfo{volume}{11}, \bibinfo{publisher}{College Pub.}, pp.
  \bibinfo{pages}{362--381}, \doi{10.48550/arXiv.1604.08709}.
\newblock \bibinfo{note}{ISBN:978-1-84890-201-5}.

\bibitemdeclare{inproceedings}{Hong}
\bibitem{Hong}
\bibinfo{author}{Bo~\surnamestart Hong\surnameend} (\bibinfo{year}{2023}):
  \emph{\bibinfo{title}{Knowing the Value of a Predicate}}.
\newblock In: {\slshape \bibinfo{booktitle}{Proceedings of LORI 2023, Lecture
  Notes in Computer Science}}, \bibinfo{volume}{14329}, pp.
  \bibinfo{pages}{149--166}, \doi{10.1007/978-3-031-45558-2_12}.

\bibitemdeclare{article}{Plaza}
\bibitem{Plaza}
\bibinfo{author}{Jan \surnamestart Plaza\surnameend}:
  \emph{\bibinfo{title}{Logics of Public Communication}}.
\newblock {\slshape \bibinfo{journal}{Synthese}} \bibinfo{volume}{158}, pp.
  \bibinfo{pages}{165--179}, \doi{10.1007/s11229-007-9168-7}.
\newblock \bibinfo{note}{Republication of J. Plaza (1989), Proceedings 4th
  International Symposium on Methodologies for Intelligent Systems,
  pp.201-216}.

\bibitemdeclare{book}{LDII}
\bibitem{LDII}
\bibinfo{author}{Johan \surnamestart {van Benthem}\surnameend}
  (\bibinfo{year}{2011}): \emph{\bibinfo{title}{Logical Dynamics of Information
  and Interaction}}.
\newblock \bibinfo{publisher}{Cambridge University Press, UK},
  \doi{10.1017/CBO9780511974533}.

\bibitemdeclare{article}{BenthemMinica1}
\bibitem{BenthemMinica1}
\bibinfo{author}{Johan \surnamestart {van Benthem}\surnameend} \&
  \bibinfo{author}{{\c{S}}tefan \surnamestart Minic{\u{a}}\surnameend}
  (\bibinfo{year}{2012}): \emph{\bibinfo{title}{Toward a Dynamic Logic of
  Questions}}.
\newblock {\slshape \bibinfo{journal}{Journal of Philosophical Logic}}
  \bibinfo{volume}{41}, pp. \bibinfo{pages}{633--669},
  \doi{10.1007/s10992-012-9233-7}.

\bibitemdeclare{incollection}{Yanjing}
\bibitem{Yanjing}
\bibinfo{author}{Yanjing \surnamestart Wang\surnameend} (\bibinfo{year}{2018}):
  \emph{\bibinfo{title}{Beyond Knowing That: A New Generation of Epistemic
  Logics}}.
\newblock In: {\slshape \bibinfo{booktitle}{Outstanding Contributions to
  Logic}}, \bibinfo{volume}{12}, \bibinfo{publisher}{Springer}, pp.
  \bibinfo{pages}{499--533}, \doi{10.1007/978-3-319-62864-6_21}.

\bibitemdeclare{inproceedings}{WangFan1}
\bibitem{WangFan1}
\bibinfo{author}{Yanjing \surnamestart Wang\surnameend} \& \bibinfo{author}{Jie
  \surnamestart Fan\surnameend} (\bibinfo{year}{2013}):
  \emph{\bibinfo{title}{Knowing that, Knowing what, and Public Communication:
  Public Announcement Logic with Kv Operators}}.
\newblock In: {\slshape \bibinfo{booktitle}{Proceedings of IJCAI 2013}},
  \bibinfo{publisher}{AAAI Press}, pp. \bibinfo{pages}{1147--1154},
  \doi{10.5555/2540128.2540293}.

\bibitemdeclare{incollection}{WangFan2}
\bibitem{WangFan2}
\bibinfo{author}{Yanjing \surnamestart Wang\surnameend} \& \bibinfo{author}{Jie
  \surnamestart Fan\surnameend} (\bibinfo{year}{2014}):
  \emph{\bibinfo{title}{Conditionally knowing what}}.
\newblock In: {\slshape \bibinfo{booktitle}{Advances in Modal Logic}},
  \bibinfo{volume}{10}, \bibinfo{publisher}{College Publications}, pp.
  \bibinfo{pages}{569--587}.
\newblock \bibinfo{note}{ISBN:978-1-84890-151-3}.

\end{thebibliography}
\end{document}